\newif\ifpictures
\picturestrue

\newif\ifcomment
\commenttrue

\documentclass[12pt]{amsart}
\usepackage{latex_base}
\usepackage{macros}

\usepackage[utf8x]{inputenc}
\usepackage{mathtools}

\usepackage{float}

\newcommand{\cV}{\mathcal{V}}

\newcommand{\alp}{\alpha}
\newcommand{\lam}{\lambda}

\newcommand{\constraints}{\ensuremath{\mathcal{G}}}

\newcommand{\feasible}[1][\constraints]{\ensuremath{#1_{+}}}

\newcommand{\NNcone}[1]{\ensuremath{\mathcal{K}({#1})}}

\newcommand{\preprime}[1]{\ensuremath{\operatorname{Prep}(#1)}}

\newcommand{\HierPut}[3]{\ensuremath{\operatorname{Hier}_{#2}^{#3}({#1})}}

\newcommand{\HierSch}[3]{\ensuremath{\operatorname{Hier}_{\preprime{#2}}^{#3}({#1})}}

\newcommand{\Prog}[3]{\ensuremath{\mathbb{#1}_{#2}^{#3}}}

\newcommand{\ProgSch}[3]{\ensuremath{\mathbb{#1}_{\preprime{#2}}^{#3}}}

\newcommand{\BoundOptimal}[3]{\ensuremath{{#1}_{#2}^{#3}}}
\newcommand{\Bound}[4]{\ensuremath{{#1}_{#2, #3}^{#4}}}

\newcommand{\BoundSch}[4]{\ensuremath{{#1}_{#2, \preprime{#3}}^{#4}}}

\newcommand{\SizePrep}[1]{\ensuremath{m_{#1}}}

\newcommand{\KG}{\mathcal{K}_\mathcal{G}}

\newcommand{\sSOS}{\mathbb{SOS}}
\newcommand{\sSOSD}{\overline{\mathbb{SOS}}}
\newcommand{\PSOS}{\mathbb{SOS}}
\newcommand{\PSOSD}{\overline{\mathbb{SOS}}}
\newcommand{\TSOS}{\textsc{SOS}}

\newcommand{\BSOS}[3]{\ensuremath{{#1}_{\PSOS, #2}^{#3}}}
\newcommand{\BSOSD}[3]{\ensuremath{{\overline{#1}}_{\PSOS, #2}^{#3}}}

\newcommand{\sSDSOS}{\ensuremath{\mathbb{SDSOS}}} 
\newcommand{\sSDSOSD}{\ensuremath{\overline{\mathbb{SDSOS}}}} 
\newcommand{\PSDSOS}{\ensuremath{\mathbb{SDSOS}}}
\newcommand{\PSDSOSD}{\overline{\mathbb{SDSOS}}}
\newcommand{\TSDSOS}{\textsc{SDSOS}}

\newcommand{\BSDSOS}[3]{\ensuremath{{#1}_{\PSDSOS, #2}^{#3}}}
\newcommand{\BSDSOSD}[3]{\ensuremath{{\overline{#1}}_{\PSDSOS, #2}^{#3}}}
\newcommand{\closSDSOS}{\ensuremath{\operatorname{cls}(\mathbb{SDSOS})}}

\newcommand{\sSA}{\ensuremath{\mathbb{SA}}} 
\newcommand{\sSAD}{\ensuremath{\overline{\mathbb{SA}}}}
\newcommand{\PSA}{\ensuremath{\mathbb{SA}}}
\newcommand{\PSAD}{\overline{\mathbb{SA}}}
\newcommand{\TSA}{\textsc{SA}}

\newcommand{\BSA}[3]{\ensuremath{{#1}_{\PSA, #2}^{#3}}}
\newcommand{\BSAD}[3]{\ensuremath{{\overline{#1}}_{\PSA, #2}^{#3}}}

\newcommand{\sSONC}{\ensuremath{\mathbb{SONC}}} 
\newcommand{\PSONC}{\ensuremath{\mathbb{SONC}}}

\newcommand{\TSONC}{\textsc{SONC}}

\newcommand{\BSONC}[3]{\ensuremath{{#1}_{\PSONC, #2}^{#3}}}
\newcommand{\closSONC}{\ensuremath{\operatorname{cls}(\mathbb{SONC})}}

\newcommand{\sSOCP}{\ensuremath{\mathbb{SOCP}}} 
\newcommand{\PSOCP}{\ensuremath{\mathbb{SOCP}}}

\newcommand{\BSOCP}[3]{\ensuremath{{#1}_{\PSOCP, #2}^{#3}}}

\newcommand{\sGEN}{\ensuremath{\mathbb{GEN}}}
\newcommand{\PGEN}{\mathbb{GEN}}

\title[]{New Dependencies of Hierarchies in Polynomial Optimization}

\author{ Adam Kurpisz}

\address{Adam Kurpisz, ETH Z\"urich, Department of Mathematics, R\"amistrasse 101,	8092 Z\"urich, Switzerland \medskip}

\email{adam.kurpisz@ifor.math.ethz.ch}

\author{Timo de Wolff}

\address{Timo de Wolff, Technische Universit\"at Berlin, Institut f\"ur Mathematik, Sekr. MA 6-2, Stra{\ss}e des 17.~Juni 136, 10623 Berlin,
	Germany\medskip}

\email{dewolff@math.tu-berlin.de}



\begin{document}

\begin{abstract}
    We compare four key hierarchies for solving Constrained Polynomial Optimization Problems (CPOP): \emph{Sum of Squares (SOS)}, \emph{Sum of Diagonally Dominant Polynomials (SDSOS)}, \emph{Sum of Nonnegative Circuits (SONC)}, and the \emph{Sherali Adams (SA)} hierarchies. We prove a collection of dependencies among these hierarchies both for general CPOPs and for optimization problems on the Boolean hypercube. Key results include for the general case that the (SONC) and (SOS) hierarchy are polynomially incomparable, while (SDSOS) is contained in (SONC). A direct consequence is the non-existence of a Putinar-like Positivstellensatz for SDSOS. On the Boolean hypercube, we show as a main result that Schm\"udgen-like versions of the hierarchies SDSOS$^*$, SONC$^*$, and SA$^*$ are polynomially equivalent. Moreover, we show that SA$^*$ is contained in \emph{any} Schm\"udgen-like hierarchy that provides a $O(n)$ degree bound.

	
\end{abstract}

	\maketitle

	\section{Introduction}

	A \struc{\emph{Constrained Polynomial Optimization Problem (CPOP)}}  is of the form
	\begin{eqnarray*}
	& & \min f(\mathbf{x}), \\
	& \text{ subject to} & g_1(\mathbf{x}),\ldots,g_m(\mathbf{x}) \geq 0, \label{Equ:HypercubeOptimizationProblem} 
	\end{eqnarray*}
	where $f(\mathbf{x})$ and $g_1(\mathbf{x}),\ldots,g_m(\mathbf{x})$ are $n$-variate real polynomials. Solving CPOP is a crucial nonconvex optimization problem, which lies at the core of both theoretical and applied computer science.
	A special case of CPOP is a \struc{\emph{Binary Constrained Polynomial Optimization Problem (BCPOP)}} where the polynomials $\pm(x_i^2-x_i)$ are among the polynomials defining the feasibility set. 
	Many important optimization problems belong to the BCPOP class. 
	However, solving these is NP-hard in general. 
	
	A CPOP can be equivalently seen as the problem of maximizing a real $\lambda$ such that $f-\lambda$ is nonnegative over the \struc{\emph{semialgebraic}} set defined by the polynomials $g_1(\mathbf{x}),\ldots,g_m(\mathbf{x})$. 
	This is an interesting perspective since various techniques form real algebraic geometry provide methods for certifying nonnegativity of a real polynomial over semialgebraic sets. The class of such theorems is called \struc{\emph{Positivstellens\"atze}}. 
	These theorems state that, under some assumptions, a polynomial $f$, which is positive (or nonnegative) over the feasibility set, can be expressed in a particular algebraic way. 
	Typically, this algebraic expression is a sum of nonnegative polynomials from a chosen \struc{\textit{ground set}} of nonnegative polynomials 
	multiplied by the polynomials defining the feasibility set.
	Choosing a proper ground set of nonnegative polynomials is crucial from the perspective of optimization. 
	Ideally, both testing membership in the ground set and deciding nonnegativity of a polynomial in the ground set should be efficiently doable. 
	Moreover, fixing the maximum degree of polynomials in the ground sets, used for a representation of $f$, provides a family of algorithms parameterized by an integer $d$, which gives a sequence of lower bounds for the value of CPOP. 
	If the ground set of polynomials is chosen properly, then the sequence of lower bounds converges in $d$ to the optimal value of CPOP.
	
	One of the most successful approaches for constructing theoretically efficient algorithms is the \struc{\textit{Sum of Squares (SOS)}} method~\cite{GrigorievV01,Nesterov00,parrilo00,schor87}, known as \struc{\textit{Lasserre relaxation}} \cite{Lasserre01}. 
	The method relies on Putinar's Positivstellensatz~\cite{Putinar93} using sum of squares of polynomials as the ground set. 
	Finding a degree $d$ SOS certificate for nonnegativity of $f$ can be performed by solving a \struc{\emph{semidefinite programming} (SDP)} formulation of size $\binom{n+d}{d}^{O(1)}$. 
	Finally, for every (feasible) $n$-variate hypercube optimization problem, with constraints of degree at most $d$, there exists a degree $2(n +\lceil d/2\rceil )$ SOS certificate, see e.g.,~\cite{BarakS16}. 
	
	The SOS algorithm is a frontier method in algorithm design. 
	It was used to provide the best available algorithms for a variety of combinatorial optimization problems. 
	The Lov\'{a}sz $\theta$-function~\cite{Lovasz79} for the \textsc{Independent Set} problem is implied by the SOS algorithm of degree 2. 
	Moreover, the Goemans-Williamson relaxation~\cite{GoemansW95} for the \textsc{Max Cut} problem and the Goemans-Linial relaxation for the \textsc{Sparsest Cut} problem (analyzed in~\cite{AroraRV09}) can be obtained by the SOS algorithm of degree 2 and 6, respectively. 
	SOS was also proven to be a successful method for \textsc{Maximum Constraint Satisfaction} problems (\textsc{Max CSP}). For \textsc{Max CSP}, the SOS algorithm is as powerful as any SDP relaxation of comparable size~ \cite{LeeRagSteu15}. 
	Furthermore, SOS was applied to problems in \textsc{dictionary learning}~\cite{BarakKS15,SchrammS17}, \textsc{tensor completion and decomposition}~\cite{BarakM16,HopkinsSSS16,PotechinS17}, and \textsc{robust estimation}~\cite{KothariSS18}. 
	For other applications of the SOS method see e.g.,~\cite{BarakRS11,BateniCG09,Chlamtac07,ChlamtacS08,DBLP:conf/soda/CyganGM13,VegaK07,GuruswamiS11,MagenM09,Mastrolilli17,RaghavendraT12}, and the surveys~\cite{Chla12,Laurent03,laurent09}.
	

	From a practical perspective however, solving SDP problems is known to be very time consuming.
	Moreover, from a theoretical point of view, it is an open problem whether an SDP of size $n^{O(d)}$ can be solved in time $n^{O(d)}$~\cite{ODonnell16,RaghavendraW17}. Hence, various methods have been proposed to choose different ground sets of polynomials to make a resulting problem easier to solve, but still effective. 
	
	In~\cite{AhmadiM14} Ahmadi and Majumdar propose an algorithmic framework by choosing the ground set of polynomials to be \struc{\emph{scaled diagonally-dominant polynomials (SDSOS)}}. 
	SDSOS polynomials can be seen as the binomial squares. Thus, the SDSOS algorithm is not stronger than the SOS algorithm. 
	However, searching for a degree-$d$ SDSOS certificate can be performed using \struc{\emph{Second Order Conic Program (SOCP)}} of size $\binom{n+d}{d}^{O(1)}$; see \cite{AhmadiM14}. Since, in practice, an SOCP can be solved much faster than an SDP, the algorithm attracted a lot of attention and has been used to solve problems in Robotics and Control~\cite{MajumdarAT14, Leong18,PeniP15, SootlaA16, ZhengFP18}, Option Pricing~\cite{AhmadiM14}, Power Flow~\cite{Kuang17, SalgadoSTL18}, and Discrete Geometry~\cite{DiasL16}.
	
	An alternative approach, that is a more tractable method than the SOS, was initiated by Sherali and Adams in~\cite{SheraliA90}. The technique was first introduced as a method to tighten the \struc{\emph{Linear Program (LP)}} relaxations for BCPOP problems and for such settings finding the degree $d$ certificate can be done by solving an LP of size 
	$\binom{n+d}{d}^{O(1)}$.
	The \struc{\emph{Sherali Adams (SA)}} algorithm arises from using the set of polynomials depending on at most $d$ variables, which are nonnegative on the Boolean hypercube. 
	These polynomials are called \struc{\textit{$d$-juntas}}. 
	The SA algorithm was used to construct some of the most prominent algorithms with good asymptotic running time in combinatorial optimization~\cite{ChanLRS16,LeveyR16,ThapperZ17}, logic~\cite{AtseriasM13}, and other fields of computer science.

	%
	%
	%
	
	Finally, a method independent from SOS was introduced in~\cite{Iliman:deWolff:Circuits} using \struc{\emph{Sum of Nonnegative Circuit Polynomials (SONC)}} as a ground set. These polynomials form a full dimensional cone in the cone of nonnegative polynomials, which is not contained in the SOS cone. 
	For example, the well-known Motzkin polynomial is a nonnegative circuit polynomial, but not an SOS.
	Moreover, SONCs generalize polynomials which are certified to be nonnegative via the arithmetic-geometric mean inequality \cite{Reznick:AGI}. 
	SONC certificates of degree $d$ can be computed via a convex optimization program called \textit{\struc{Relative Entropy Programming (REP)}} of size $\binom{n+d}{d}^{O(1)}$ \cite[Theorem 5.3]{Dressler:Iliman:deWolff:Positivstellensatz}; see also \cite{Chandrasekaran:Shah:SignomialOptimization,Chandrasekaran:Murray:Wierman}. Recently, an experimental comparison of SONC with the SOS method for unconstrained optimization was presented in~\cite{Seidler:deWolff:ExperimentalComparisonSONCandSOS}.

	For all presented algorithms, one can define a potentially stronger algorithm without changing the corresponding ground set of polynomials, by using a more general construction for the \struc{\textit{certificate}} of nonnegativity. Such a certificate expresses a polynomial, which is nonnegative over a given semialgebraic set, as a sum of polynomials from the ground set multiplied by the product of polynomials defining the semialgebraic set; see \cref{sec:Preliminaries} for further details. We call the resulting systems $\struc{SOS^*}$, $\struc{SDSOS^{*}}$, $\struc{SA^{*}}$ and $\struc{SONC^*}$. Some of these extensions were intensively studied in the literature, see e.g.,~\cite{GrigorievHP02,Worah15}. 
	\subsection*{Our Results}
	In this paper, we provide an extensive comparison of the presented semialgebraic proof systems. More precisely, following the definitions in e.g.,~\cite{BeameFIKPPR18}, we analyze their \struc{\textit{polynomial comparability}}:
	\begin{definition}
		Let $P$ and $Q$ be semialgebraic proof systems. $P$ \struc{\emph{contains}} $Q$  if for every semialgebraic set $\feasible$ and a polynomial $f$ admitting a degree $d$ certificate of nonnegativity over $\feasible$ in $Q$, $f$ admits also a degree $O(d)$ certificate in $P$.
		System $P$ \struc{\emph{strictly contains}} $Q$  if $P$ contains $Q$ but $Q$ does not contain $P$.
		Systems $P$ and $Q$ are \struc{\emph{polynomially equivalent}} if $P$ contains $Q$ and $Q$ contains $P$. Finally, systems $P$ and $Q$ are \struc{\emph{polynomially incomparable}} if neither $P$ contains $Q$ nor $Q$ contains $P$.
		\label{Def:PolynomialEquivalent}
	\end{definition}
	
	For a more detailed definition of proof systems and their comparability, see \cref{Sec:ComparingProofSystems}.
	
	
	In this article, we show the dependencies between the proof systems presented in \cref{fig:moti}.
	
	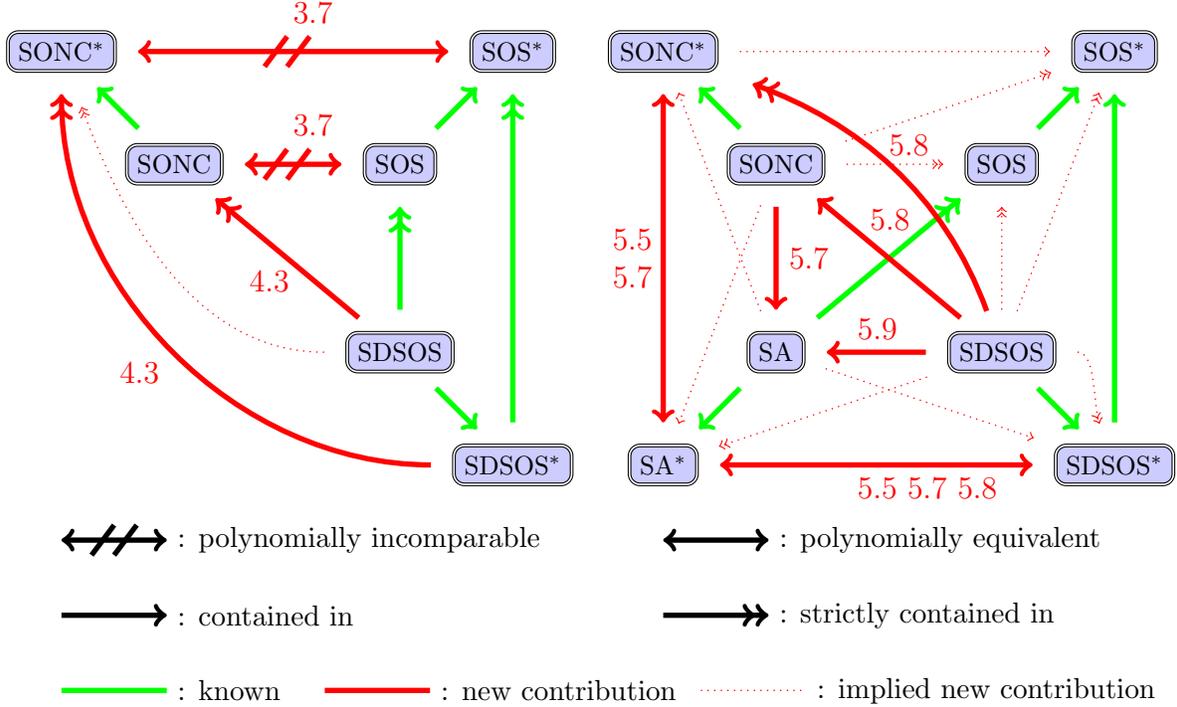
\begin{figure}[H]
		\begin{tikzpicture}
		\pgfmathsetmacro{\myheight}{2.5}
		\pgfmathsetmacro{\myshift}{0.8}

		
		\node [fill=blue!20,draw, double, rounded corners] (sonc*) at (0,5.5) {\footnotesize SONC$^*$};
		\node [fill=blue!20,draw, double, rounded corners] (sdsos*) at (6,0) {\footnotesize SDSOS$^*$};
		\node [fill=blue!20,draw, double, rounded corners] (sos*) at (6,5.5) {\footnotesize SOS$^*$};
		
		\draw[<->, line width = 2pt, red, shorten <=0.3cm, shorten >=0.3cm] (sonc*) --node [text width=0.0cm,midway,above=0.5em ]{\ref{cor:sonc_vs_sos}} (sos*) {{}};
		\draw[line width = 2pt, red] (2.7,5.3) -- (3,5.7);
		\draw[line width = 2pt, red] (3,5.3) -- (3.3,5.7);
		\draw[->>, line width = 2pt, green, shorten <=0.3cm, shorten >=0.3cm] (sdsos*) -- (sos*) {{}};
		\draw[->>, line width = 2pt, red, shorten <=0.3cm, shorten >=0.3cm] (sdsos*) to [out=180,in=270] node [text width=1.5cm,midway,below ] {\ref{cor:sdsos_vs_sonc}} (sonc*) {{}};

		\node [fill=blue!20,draw, double, rounded corners] (sonc) at (1.5,4) {\footnotesize SONC};
		\node [fill=blue!20,draw, double, rounded corners] (sdsos) at (4.5,1.5) {\footnotesize SDSOS};
		\node [fill=blue!20,draw, double, rounded corners] (sos) at (4.5,4) {\footnotesize SOS};
				
		\draw[<->, line width = 2pt, red, shorten <=0.3cm, shorten >=0.3cm] (sonc) -- node [text width=0.0cm,midway,above=0.5em ]{\ref{cor:sonc_vs_sos}} (sos) {{}};
		\draw[line width = 2pt, red] (2.7,3.8) -- (3,4.2);
		\draw[line width = 2pt, red] (3,3.8) -- (3.3,4.2);
		\draw[->>, line width = 2pt, green, shorten <=0.3cm, shorten >=0.3cm] (sdsos) -- (sos) {{}};
		\draw[->>, line width = 2pt, red, shorten <=0.3cm, shorten >=0.3cm] (sdsos) -- node [text width=1cm,midway,below ] {\ref{cor:sdsos_vs_sonc}} (sonc) {{}};

		\draw[->, line width = 2pt, green, shorten <=0.3cm, shorten >=0.3cm] (sdsos) -- (sdsos*) {{}};
		\draw[->, line width = 2pt, green, shorten <=0.3cm, shorten >=0.3cm] (sonc) -- (sonc*) {{}};
		\draw[->, line width = 2pt, green, shorten <=0.3cm, shorten >=0.3cm] (sos) -- (sos*) {{}};

		\draw[->>, dotted, line width = 0.5pt, red, shorten <=0.3cm, shorten >=0.5cm] (sdsos) to [out=180,in=290]  (sonc*) {{}};


		\node [fill=blue!20,draw, double, rounded corners] (sonc*2) at (8,5.5) {\footnotesize SONC$^*$};
		\node [fill=blue!20,draw, double, rounded corners] (sdsos*2) at (14,0) {\footnotesize SDSOS$^*$};
		\node [fill=blue!20,draw, double, rounded corners] (sos*2) at (14,5.5) {\footnotesize SOS$^*$};
		\node [fill=blue!20,draw, double, rounded corners] (sa*2) at (8,0) {\footnotesize SA$^*$};
	
		\draw[<->, line width = 2pt, red, shorten <=0.3cm, shorten >=0.3cm] (sonc*2) -- node [text width=0.5cm,midway,left ] {\ref{lem:sonc_vs_sa_0/1}~\\\ref{cor:sonc_vs_sa_0/1}} (sa*2) {{}};
		\draw[<->, line width = 2pt, red, shorten <=0.3cm, shorten >=0.3cm] (sdsos*2) -- node [text width=0.5cm,midway,below ] {\ref{lem:sonc_vs_sa_0/1}~\ref{cor:sonc_vs_sa_0/1}~\ref{thm:sonc_vs_sdsos_0/1}} (sa*2) {{}};
		\draw[->, line width = 2pt, green, shorten <=0.3cm, shorten >=0.3cm] (sdsos*2) -- (sos*2) {{}};
		
		\node [fill=blue!20,draw, double, rounded corners] (sonc2) at (9.5,4) {\footnotesize SONC};
		\node [fill=blue!20,draw, double, rounded corners] (sdsos2) at (12.5,1.5) {\footnotesize SDSOS};
		\node [fill=blue!20,draw, double, rounded corners] (sos2) at (12.5,4) {\footnotesize SOS};
		\node [fill=blue!20,draw, double, rounded corners] (sa2) at (9.5,1.5) {\footnotesize SA};

		\draw[->, line width = 2pt, red, shorten <=0.3cm, shorten >=0.3cm] (sonc2) -- node [text width=0.5cm,midway,right] {\ref{cor:sonc_vs_sa_0/1}}(sa2) {{}};
		\draw[->, line width = 2pt, red, shorten <=0.3cm, shorten >=0.3cm] (sdsos2) -- node [text width=0.5cm,midway,above ] {\ref{thm:SDSoSvsSA_0/1}} (sa2) {{}};
		\draw[->>, line width = 2pt, green, shorten <=0.3cm, shorten >=0.3cm] (sa2) -- (sos2) {{}};
		\draw[->, line width = 2pt, red, shorten <=0.3cm, shorten >=0.3cm] (sdsos2) -- node [text width=0.5cm,midway,above=0.5em ] {\ref{thm:sonc_vs_sdsos_0/1}}  (sonc2) {{}};		
								
		\draw[->, line width = 2pt, green, shorten <=0.3cm, shorten >=0.3cm] (sdsos2) -- (sdsos*2) {{}};
		\draw[->, line width = 2pt, green, shorten <=0.3cm, shorten >=0.3cm] (sonc2) -- (sonc*2) {{}};
		\draw[->, line width = 2pt, green, shorten <=0.3cm, shorten >=0.3cm] (sos2) -- (sos*2) {{}};
		\draw[->, line width = 2pt, green, shorten <=0.3cm, shorten >=0.3cm] (sa2) -- (sa*2) {{}};	
		\draw[->>, line width = 2pt, red, shorten <=0.3cm, shorten >=0.5cm] (sdsos2) to [out=110,in=340]  node [text width=0.0cm,midway,above ] {\ref{thm:sonc_vs_sdsos_0/1}} (sonc*2) {{}};	

		\draw[->>, dotted, line width = 0.5pt, red, shorten <=0.3cm, shorten >=0.3cm] (sonc2) -- (sos2) {{}};
		\draw[->, dotted, line width = 0.5pt, red, shorten <=0.3cm, shorten >=0.3cm] (sa2) -- (sonc*2) {{}};
		\draw[->, dotted, line width = 0.5pt, red, shorten <=0.3cm, shorten >=0.3cm] (sonc2) -- (sa*2) {{}};
		\draw[->, dotted, line width = 0.5pt, red, shorten <=0.3cm, shorten >=0.3cm] (sa2) -- (sdsos*2) {{}};		
		\draw[->>, dotted, line width = 0.5pt, red, shorten <=0.3cm, shorten >=0.3cm] (sdsos2) -- (sa*2) {{}};		
		\draw[->>, dotted, line width = 0.5pt, red, shorten <=0.3cm, shorten >=0.3cm] (sonc2) -- (sos*2) {{}};			
		\draw[->, dotted, line width = 0.5pt, red, shorten <=0.3cm, shorten >=0.3cm] (sonc*2) -- (sos*2) {{}};						
		
		\draw[->>, dotted, line width = 0.5pt, red, shorten <=0.3cm, shorten >=0.3cm] (sdsos2) -- (sos*2) {{}};									
		\draw[->>, dotted, line width = 0.5pt, red, shorten <=0.3cm, shorten >=0.3cm] (sdsos2) -- (sos2) {{}};
		
		\draw[->>, dotted, line width = 0.5pt, red, shorten <=0.3cm, shorten >=0.3cm] (sdsos2) to [out=0,in=110]  
		(sdsos*2) {{}};
		
		\draw[<->, line width = 2pt] (0,-1) -- (1.4,-1) {{}};
		\draw[line width = 2pt] (0.4,-1.2) -- (0.7,-0.8);
		\draw[line width = 2pt] (0.7,-1.2) -- (1,-0.8);
		\node[right] at (1.4,-1) {\small : polynomially incomparable};
		
		\draw[<->, line width = 2pt] (8,-1) -- (9.4,-1) {{}};
		\node[right] at (9.4,-1) {\small : polynomially equivalent};
		
		\draw[->, line width = 2pt] (0,-2) -- (1.4,-2) {{}};
		\node[right] at (1.4,-2) {\small : contained in};
	
		\draw[->>, line width = 2pt] (8,-2) -- (9.4,-2) {{}};
		\node[right] at (9.4,-2) {\small : strictly contained in};	
	
		\draw[line width = 2pt, green] (0,-3) -- (1.4,-3) {{}};
		\node[right] at (1.4,-3) {\small : known};
		
		\draw[line width = 2pt, red] (3.5,-3) -- (4.9,-3) {{}};
		\node[right] at (4.9,-3) {\small : new contribution};
		
		\draw[line width = 0.5pt, dotted, red] (8.5,-3) -- (9.9,-3) {{}};
		\node[right] at (9.9,-3) {\small : implied new contribution};
		\end{tikzpicture}
		
		\caption{A visualization of our results. Left hand side: CPOP. Right hand side: BCPOP. Labels on the arrows refer to theorems.}\label{fig:moti}
	\end{figure}

	In particular, in \cref{Sec:SONCvsSOS}, we show that for general CPOP problems the SOS proof system is polynomially incomparable with the SONC proof system. We also proved that the same relation holds for SOS$^*$ and SONC$^*$ proof systems;
	see \cref{cor:sonc_vs_sos}. So far, it was only known that the cones of SOS and SONC polynomials are not contained in each other \cite[Proposition 7.2]{Iliman:deWolff:Circuits} however, it has no direct implication on the relation between the SOS and the SONC methods for the CPOP optimization. Similarly, in a very recent result \cite{Chandrasekaran:Murray:Wierman}, the authors point out that the SONC cone contains SDSOS cone. In this paper, in \cref{Sec:SDSOSvsSONC}, we extend this result for CPOP problems by proving, that SONC certificate strictly contains the SDSOS certificate and the same relation holds for SONC$^*$ and SDSOS$^*$ certificates; see \cref{cor:sdsos_vs_sonc}. As a consequence, we conclude that there exists no Putinar-like Positivstellensatz for SDSOS; see \cref{Corollary:ThereExistsNoPutinarPositivstellensatzForSDSOS}. 
	
	For the BCPOP we provide a general, sufficient condition for the proof system to contain $SA^*$ proof system, see \cref{thm:General_proof_as_strong_as_SA}. This combined with the results from \cref{sec:sonc_vs_sa_0/1}, and \cref{sec:sonc_vs_sdsos_0/1} proves the polynomial equivalence of $SONC^*$, $SDSOS^*$, and $SA^*$ on the Boolean hypercube. Moreover, by proving some properties of SONC, SDSOS, and SA polynomials in \cref{lem:f_sdoss_is_sonc}, and~\cref{Lemma:CircuitsAreJuntas}, we prove additional dependencies between the hierarchies in \cref{sec:sonc_vs_sa_0/1}, \cref{sec:sonc_vs_sdsos_0/1}, and~\cref{sec:sa_vs_sdsos_0/1}.
	
	We remark that all results in this article concern the \textit{minimal degrees} for certificates in a particular proof system as these are the standard way to measure the complexity of algorithms in theoretical computer science. Our results do not directly imply a particular behaviour of actual runtimes in an experimental setting, as these depend on various further factors other than the degree.

	\subsection*{Acknowledgements}
	AK is supported by SNSF project PZ00P2$\_$174117,
	TdW is supported by the DFG grant WO 2206/1-1.
	
	\section{Preliminaries}
	\label{sec:Preliminaries}

	In this section, we introduce the proof systems used in this article. 
	Moreover, for the sake of clarity, we provide dual formulations for some of the presented proof systems for the BCPOP case.
	We begin with introducing basic notation.
	For any $n, d \in \N$ we denote $\struc{[n]} = \{1,\ldots,n\}$ and $\struc{\binom{n}{\leq d}}~:=~\sum_{i=0}^d \binom{n}{i}$. 
	Let $\struc{\N^*} = \N \setminus \{\mathbf{0}\}$ and $\struc{\R_{\geq 0}}$ ($\struc{\R_{> 0}}$) be the set of nonnegative (positive) real numbers. 
	Let $\struc{\mathbb{R}[\Vector{x}]} = \R[x_1,\ldots,x_n]$ be the ring of \textit{\struc{$n$-variate real polynomials}} and for every $f \in \R[\Vector{x}]$ we define the \textit{\struc{real zero set}} as $\struc{\cV(f)} = \{\Vector{x} \in \R^n \ | \ f(\Vector{x}) = 0\}$. 
	We denote the \textit{\struc{Newton polytope}} of $f$ by $\struc{\New(f)}$ and the \struc{\textit{vertices}} of $\New(f)$ by $\struc{\vertices{\New(f)}}$. 
	A lattice point is called \struc{\textit{even}} if it is in $(2\N)^n$, and a term $ f_{\boldsymbol{\alp}}\mathbf{x}^{\boldsymbol{\alp}}$ is called a \struc{\emph{monomial square}} if $f_{\boldsymbol{\alp}} \in \R_{\geq 0}$ and $\alp$ is even. 
	
	In what follows we introduce different proof systems and their notation. Next to the specific sources that we provide later in the section, we refer the reader to introductory literature like \cite{Blekherman:Parrilo:Thomas,laurent09,Lasserre:Book:IntroductionPolynomialOptimization,Marshall:Book} on the mathematical side, and~\cite{Razborov16,Rothvoss13} on the computer science side. Moreover, we fix the notation 
	\begin{align*}
	\struc{\constraints} \ := \ \{g_0:=1,g_1,\ldots,g_m:~g_i \in \mathbb{R}[\Vector{x}] \text{ for all } ~ i \in [m] \}
	\end{align*}
	for a set of polynomials. Throughout the paper we assume that the cardinality $\struc{m}$ of the set $\constraints$ is polynomial in the size of $n$. 
	For a given $\constraints$, we define the corresponding \textit{\struc{semi-algebraic set}}
	\begin{align*}
	\struc{\feasible} \ := \ \{ \Vector{x} \in \mathbb{R}^n~ |~ g(\Vector{x})\geq 0 \text{ for all } {g \in \constraints}  \} \subseteq \R^n.
	\end{align*}
	Furthermore, for any given semialgebraic set $\feasible \subseteq \R^n$, we consider the set of \struc{\textit{nonnegative polynomials}} with respect to $\feasible$
	\begin{align*}
	\struc{\NNcone{\feasible}} \ := \ \{f \in \R[\Vector{x}] \ | \ f(\Vector{x}) \geq 0 \text{ for all } \Vector{x} \in \feasible\}.
	\end{align*}
	For a given $f \in \R[\Vector{x}]$ and a set of constraints $\constraints$, we define the corresponding \struc{\textit{constrained polynomial optimization problem} (CPOP)} as (see e.g.,~\cite{Boyd04})
	\begin{align}
	\label{eq:intro_PP_2}
	\begin{aligned}
	\struc{\BoundOptimal{f}{\constraints}{*}} \ := \ \min\{f(\Vector{x})~ |~ \Vector{x} \in \feasible\} \ = \ \max\{\lambda \in \mathbb{R}~ |~ f-\lambda \in \NNcone{\feasible}\}.
	\end{aligned}
	\tag{CPOP}
	\end{align}
	Hence, $\feasible$ corresponds to the feasibility region of the program \eqref{eq:intro_PP_2}.
	
	The problem \eqref{eq:intro_PP_2} is NP-hard in general. Thus, one chooses proper subsets $\NNcone{\feasible}^{'} \subseteq \NNcone{\feasible}$ such that, on the one hand, the corresponding polynomial optimization problem provides a lower bound on the value of~\eqref{eq:intro_PP_2} and on the other hand, is computationally \emph{tractable}. Such subsets are called \textit{\struc{certificates of nonnegativity}}. The choice of a suitable certificate of nonnegativity is crucial for obtaining a good lower bound for the problem~\eqref{eq:intro_PP_2}. 
	
	\medskip
	
	Let us be more specific. For a given $\constraints$ the induced \textit{\struc{preprime}} is given by
	\begin{align*}
	\struc{\preprime{\constraints}} \ := \ \left\{\sum_{i = 0}^s c_i g_0^{k_{i,0}} \cdots g_m^{k_{i,m}} \ \middle| \ c_i \in \R_{\geq 0}, g_j \in \constraints, k_{i,j} \in \N\right\}.
	\end{align*} 
	Note that $\NNcone{\feasible} = \NNcone{\preprime{\constraints}_+}$.
	Throughout the paper we assume that for a given $\constraints$ $\struc{\SizePrep{2d}}$ is the cardinality of the set $\preprime{\constraints}$ restricted to polynomials of degree at most $2d$.
	In order to relax \eqref{eq:intro_PP_2} to a finite size optimization problem we introduce \struc{\textit{polynomial hierarchies}}.
	\begin{definition}
		\label{def:hierarchy_of_certificates}
		Let $\constraints$ be a collection of polynomials and let $\struc{\sGEN}$ be a subset of $\NNcone{\feasible}$. We define the following degree $d \in \N$  depending \struc{\textit{hierarchy of certificates of nonnegativity}}: 
		\begin{align*}
		\small
		\struc{\HierPut{\sGEN}{\constraints}{2d}} \ := \ \left\{\sum_{i = 1}^m s_i g_i \subseteq \R[\Vector{x}] \ \middle| \ s_i \in \sGEN, \ g_i \in \constraints, \text{ and }\deg(s_ig_i) \leq 2d \text{ for all } i \in [m] \right\}.
		\end{align*}
		In several contexts it is more useful to consider the preprime of the constraints, i.e., $\struc{\HierSch{\sGEN}{\constraints}{2d}}$.
		Every such hierarchy of polynomials yields a sequence of lower bounds given by the following optimization program:
		\begin{align}
		\label{eq:OptimizationProblemGeneralHierarchy}
		\begin{aligned}
		\struc{\Bound{f}{\sGEN}{\constraints}{2d}} \ := \ \max \left\{ \lambda \in \mathbb{R}~ |~ f-\lambda \in \HierPut{\sGEN}{\constraints}{2d} \right\} , \text{ and }
		\end{aligned}
		\tag{$\Prog{\sGEN}{\constraints}{2d}$}
		\end{align}
		\begin{align}
		\label{eq:OptimizationProblemGeneralHierarchySchmuedgen}
		\begin{aligned}
		\struc{\BoundSch{f}{\sGEN}{\constraints}{2d}} \ := \ \max\left\{\lambda \in \mathbb{R}~ |~ f-\lambda \in \HierSch{\sGEN}{\constraints}{2d}\right\} \text{ respectively.}
		\end{aligned}
		\tag{$\ProgSch{\sGEN}{\constraints}{2d}$}
		\end{align}
		\label{Def:PolynomialHierarchy}
	\end{definition}
	
	Throughout this paper we assume that the set $\mathcal{G}$ is chosen such that $\preprime{G}$ is \struc{\emph{Archimedean}}, a property which is e.g., implied by the compactness of $\mathcal{G}_+$. In what follows we occasionally enforce compactness of $\preprime{G}$ by adding \struc{\textit{box constraints}} $\struc{l_i} := x_i \pm N \geq 0$ to $\constraints$ with $N \in \N$ sufficiently large for $i \in [n]$.
	
	Under this assumption we obtain from Krivine's general Positivstellensatz \cite{Krivine:Positivstellensatz1,Krivine:Positivstellensatz2}, see also \cite[Theorem 5.4.4]{Marshall:Book}, the following \struc{\textit{Schm\"udgen-type Positivstellensatz}}; see \cite[Theorem 5.1]{Schweighofer:SchmuedgensPositivstellensatz}:
	
	\begin{theorem}
		\label{thm:Schmudgen}
		Let $\preprime{\constraints}$ be Archimedean and let $\sGEN \subseteq \NNcone{\feasible}$ such that $\sGEN$ is closed under addition. Let $f(\Vector{x}) > 0$ for all $\Vector{x} \in \feasible$. Then there exists a $d \in \N$ such that $\BoundSch{f}{\sGEN}{\constraints}{2d} = \BoundOptimal{f}{\constraints}{*}$.
	\end{theorem}
	
	For the SOS hierarchy this theorem was first shown by Schm\"udgen in \cite{Schmuedgen:Positivstellensatz}.

	In the following subsections we introduce some of the most prominent inner approximations of the cone $\NNcone{\feasible}$.
	
	\subsection{Sum of Squares}
	
	The \textit{\struc{SOS method}} approximates the cone $\NNcone{\feasible}$ by using the set of \textit{\struc{sum of square polynomials}} instead of the entire set of nonnegative polynomials. 
	Let $\struc{\sSOS} \ :=\ \{s~|~s=\sum_{i = 1}^k f_i^2,~f \in \mathbb{R}[\Vector{x}], ~k \in \N^* \}$ be the set of \struc{\textit{(finite) sum of square polynomials (SOS)}}. 
	The \struc{\textit{SOS program of degree $2d$}} takes the following form:
	\begin{align}
	\label{eq:intro_PP_4}
	\begin{aligned}
	\struc{\BSOS{f}{\mathcal{G}}{2d}} \ := \ \quad \max\{\lambda \in \mathbb{R}~ |~ f-\lambda \in \HierPut{\sSOS}{\constraints}{2d}\};
	\end{aligned}
	\tag{$\Prog{\sSOS}{\constraints}{2d}$}
	\end{align}
	analogously the \struc{\textit{SOS$^{*}$ program of degree $2d$}} takes the form $\struc{\ProgSch{\sSOS}{\constraints}{2d}}$.
	For the SOS-hierarchy Putinar proved the following \emph{Positivstellensatz}, which is an improvement of Schm\"udgen's Positivstellensatz.
	\begin{theorem}[Putinar's Positivstellensatz; \cite{Putinar93}]
		\label{thm:intro_Putinar}
		Let $\constraints$ be a set of polynomial constraints with $\bigcup_{d \in \N}\HierPut{\sSOS}{\constraints}{2d}$ being Archimedean, and let $f \in \R[\Vector{x}]$ with $f(\Vector{x})>0$ for all $\Vector{x} \in \feasible$. Then there exists a $d \in \N$ such that $\Bound{f}{\PSOS}{\constraints}{2d} = \BoundOptimal{f}{\constraints}{*}$.		
	\end{theorem}

	\cref{thm:intro_Putinar} provides a sequence of cones that approximate $\KG$ from the inside, such that the values of $\BSOS{f}{\mathcal{G}}{2d}$
	give a sequence of lower bounds that converges in $d$ to the optimal value of~\eqref{eq:intro_PP_2}.

	The program~\eqref{eq:intro_PP_4} can be solved using a \struc{\emph{semidefinite program} (SDP)} of size $\binom{n+d}{d}^{O(1)}$; see e.g.,~\cite{Lasserre01, Nesterov00, parrilo00, schor87}. 
	This is implied by the following fact; see e.g.,~\cite{parrilo00}.
	\begin{theorem}
		\label{thm:intro_sos_sdp}
		A polynomial $p \in \mathbb{R}[\Vector{x}]$ is a SOS of degree $d$ if and only if there exists a positive semidefinite matrix $G$, called the \textit{\struc{Gram matrix}}, such that $p=\Vector{z}^\top G \Vector{z}$, for $\Vector{z}$ being the vector of $n$-variate monomials of total degree at most $d$.
	\end{theorem}
	
	The size of the SDP program is $\binom{n+d}{d}^{O(1)}$. Moreover, for BCPOP problems, when \struc{\textit{hypercube constraints} $\pm(x_i^2- x_i)$} are incorporated in $\mathcal{G}$, it is known for $\struc{d_\mathcal{G}}:=\max\{\deg(g)~|~g \in \mathcal{G} \}$ that  $\PSOS^{2n+2\lceil d_\mathcal{G}/2 \rceil}$ solves the problem exactly, i.e., $\BSOS{f}{\mathcal{G}}{2n+2\lceil d_\mathcal{G}/2 \rceil}=\BoundOptimal{f}{\constraints}{*}$; see e.g.,~\cite{BarakS14}.

	\subsubsection{SOS - The dual perspective: Lasserre hierarchy}
	\label{sec:sos_dual}
	
	Consider a BCPOP. Let $\lambda \in \R$ be such that $f-\lambda \notin  \HierPut{\sSOS}{\constraints}{2d}$, for some $d \in \N^{*}$. By the hyperplane separation theorem for convex cones, there exists a hyperplane that separates $f-\lambda$ from $ \HierPut{\sSOS}{\constraints}{2d}$. Note that for BCPOP we can restrict to polynomials defined on the hypercube $\{0,1\}^n \to \R$, i.e., to the vector space of multi-linear polynomials. The hyperplane is represented by the polynomial $\struc{\mu}:\{0,1\}^n \to \R$, which is a normal vector to the hyperplane, such that for every polynomial $h \in \HierPut{\sSOS}{\constraints}{2d}$ we have $\sum_{\Vector{x}\in \{0,1\}^n}\mu(\Vector{x})\cdot h(\Vector{x})\geq 0$ and $\sum_{\Vector{x}\in \{0,1\}^n}(f-\lambda)(\Vector{x})\cdot h(\Vector{x}) < 0$. By scaling we can assume that $\sum_{\Vector{x}\in \{0,1\}^n}\mu(\Vector{x})=1$. 
	To every function $\mu$ we can associate a linear operator $\struc{\tilde{\mathbb{E}}_\mu}: \{0,1\}^n \to \R$ mapping polynomials to real numbers, defined by 
	\begin{align*}
		\tilde{\mathbb{E}}_\mu[h] \ := \ \sum_{\Vector{x}\in \{0,1\}^n}\mu(\Vector{x})\cdot h(\Vector{x}),
	\end{align*}
	 which is called the \struc{\emph{pseudoexpectation}}.
	 The dual problem to~\eqref{eq:intro_PP_4} is the \struc{$\overline{\text{SOS}}$ \textit{program of degree} $2d$}. It takes the form
	 \begin{align}
	 \label{eq:intro_PP_5}
	 {\small\begin{aligned}
	 \struc{\BSOSD{f}{\mathcal{G}}{2d}} \ := \ \quad \min\left\{\tilde{\mathbb{E}}_\mu[f]~ \middle|~ \mu:\{0,1\}^n\to \R,~ \tilde{\mathbb{E}}_\mu[1] = 1,~ \tilde{\mathbb{E}}_\mu[h] \geq 0,~\text{for all } h \in \HierPut{\sSOS}{\constraints}{2d}\right\}  
	 \end{aligned}},
	 \tag{$\Prog{\sSOSD}{\constraints}{2d}$}
	 \end{align}
	and is known as the \struc{\emph{Lasserre relaxation (of degree $2d$)} }.  It can be solved using an SDP of size $\binom{n}{\leq d}^{O(1)}$~\cite{Lasserre01}. Analogously, the \struc{$\overline{\text{SOS}}^{~*}$ \textit{program of degree} $2d$} takes the form $\struc{\ProgSch{\sSOSD}{\constraints}{2d}}$.

	Problem~\eqref{eq:intro_PP_5} can be reformulated in terms of \struc{\emph{moments /  localizing matrices}}. Consider $h= s^2\cdot g \in \HierPut{\sSOS}{\constraints}{2d}$, for $s \in \mathbb{R}[\Vector{x}]$ and $g \in \mathcal{G}$. Let $\struc{d^{'}}= \lfloor \frac{2d-\deg(g)}{2} \rfloor$  and $\struc{\Vector{v}}$ be the \textit{\struc{vector of coefficients}} of $h$, such that $h(\Vector{x})=\sum_{I \subseteq [n]} v_I \prod_{i \in I} x_i$.
	We can write 
	{\small
	\begin{equation}
	\label{eq:pseudoexpectation_to_matrix}
	 0\leq \tilde{\mathbb{E}}_\mu[h^2]= 
		\tilde{\mathbb{E}}_\mu\left[\left( \sum_{\substack{I,J \subseteq [n] \\ |I|,|J| \leq d^{'}}} v_I v_J\prod_{k \in I } x_k \prod_{k \in J} x_k  \right) g \right] =
		\sum_{\substack{I,J \subseteq [n] \\ |I|,|J| \leq d^{'}}} v_I v_J \tilde{\mathbb{E}}_\mu\left[g \cdot\prod_{k \in I \cup J} x_k   \right] =
		\Vector{v}^\top M_g^{2d} \Vector{v},
	\end{equation}}
	where $\struc{M_g^{2d}} \in \R^{\binom{n}{\leq d^{'}} \times\binom{n}{\leq d^{'}}}$ is a real, symmetric matrix whose rows and columns are indexed	
	 by sets $I,J \subseteq [n]$ of size at most $d^{'}$ such that $M_g^{2d}(I,J)=\tilde{\mathbb{E}}_\mu[g \cdot\prod_{k \in I \cup J} x_k]$. For $g=1$ the matrix is called the \struc{\emph{moment matrix}}, and for all other $g$ it is called the \struc{\emph{localizing matrix}} for the constraint $g$. Since for every real valued vector $\Vector{v}$ the requirement $\Vector{v}^\top M \Vector{v} \geq 0$ is equivalent to $M$ being positive semidefinite (PSD), denoted by $\struc{M \succeq 0}$, we can reformulate~\eqref{eq:intro_PP_5} as:
	 \begin{align}
	 \label{eq:intro_PP_6}
	 \begin{aligned}
	\struc{\BSOSD{f}{\mathcal{G}}{2d}} \ := \ \quad \min \left\{\tilde{\mathbb{E}}_\mu[f]~ \middle|~ \mu:\{0,1\}^n\to \R~\tilde{\mathbb{E}}_\mu[1] = 1,~ M_g^{2d} \succeq 0,~ \text{for all } g \in \mathcal{G}\right\}.
	\end{aligned}
	\end{align}

%
%

	\subsection{Scaled Diagonally Dominant Sum of Squares}
	%
	%
	%

	In~\cite{AhmadiM14} Ahmadi and Majumdar proposed an approximation of the cone $\sSOS$ based on \struc{\emph{scaled diagonally-dominant polynomials} ($\TSDSOS$)}, defined below in \cref{sec:sdd_polynomials}.
	Let $\struc{\sSDSOS}$ be the set of finite sums of scaled diagonally-dominant polynomials.
	We obtain the following program:
	\begin{align}
	\label{eq:intro_PP_7}
	\begin{aligned}
	\struc{\BSDSOS{f}{\mathcal{G}}{2d}} \ := \ \max\{\lambda \in \mathbb{R}~ |~ f-\lambda \in \HierPut{\sSDSOS}{\constraints}{2d}\};
	\end{aligned}
	\tag{$\Prog{\PSDSOS}{\constraints}{2d}$}
	\end{align}
	analogously, the \struc{\textit{SDSOS$^{~*}$}} program takes the form $\ProgSch{\PSDSOS}{\constraints}{2d}$. Since $\HierPut{\sSDSOS}{\constraints}{2d} \subseteq \HierPut{\sSOS}{\constraints}{2d} \subseteq \KG$ for every $d \in \N$, we have $\BSDSOS{f}{\mathcal{G}}{2d} \leq \BSOS{f}{\mathcal{G}}{2d}~\leq~\BoundOptimal{f}{\constraints}{*}$.
	Moreover, $\PSDSOS^{2d}$ can be solved using \struc{\emph{Second Order Conic Programming} (SOCP)} of size $\binom{n+d}{d}^{O(1)}$; see \cite{AhmadiM14}.
	

\subsubsection{Scaled diagonally-dominant polynomials}
\label{sec:sdd_polynomials}
We introduce the formal details for SDSOS certificates.
\begin{definition}
	\label{def:intro_sdsos}
	A real symmetric $m \times m$ matrix $M$ is called \struc{\emph{diagonally-dominant} (dd)} if for every $i \in [m]$ we have $M(i,i)\geq \sum_{i\neq j} |M(i,j)|$. Moreover, $M$ is called \struc{\emph{scaled diagonally-dominant} (sdd)} if there exist a positive real diagonal matrix $D$ such that $DMD$ is dd.
	A polynomial $p(\Vector{x}) \in \mathbb{R}[\Vector{x}]$ of total degree $d$ is \struc{\emph{scaled diagonally-dominant}}, denoted \struc{$p \in \sSDSOS$}, if there exist an sdd matrix $M$ such that $p=\Vector{z}^\top M \Vector{z}$, for $\Vector{z}$ being the vector of $n$-variate monomials of total degree at most $d$.
\end{definition}

Every $\TSDSOS$ polynomials is an $\TSOS$ polynomial: By \cref{def:intro_sdsos}, consider an sdd matrix $DMD$, for $M$ being a dd matrix. By the Gershgorin circle theorem, the matrix $M$ is PSD. Moreover, $DMD=DMD^{\top}$. Since $DMD^{\top}$ is a congruent transformation of $M$, that does not change the sign of the eigenvalues, the matrix $DMD^{\top}$ is also PSD.

%
%
%

Next, we provide a further characterization of SDSOS polynomials. We start with recalling the known characterization of diagonally dominant (dd) matrices.
\begin{lemma}[\cite{Barker1975}]
	\label{lem:dd_matrices}
	A symmetric $m \times m$ matrix $M$ is dd if and only if 
	\begin{align*}
		M \ = \ \sum_{i,j \in [m]}c_{ij}\Vector{w}_{ij} \Vector{w}_{ij}^\top
	\end{align*}
	for $c_{ij}\geq0$ and $\{\Vector{w}_{ij} \}_{i,j \in [m]} \subseteq \mathbb{R}^n$ being a set of vectors, each with at most two nonzero entries at positions $i$ and $j$ which equal $\pm 1$.
\end{lemma}

By \cref{def:intro_sdsos} and \cref{lem:dd_matrices}, every $n$-variate sdd polynomial $s$ of degree at most $d$ is of the form
\begin{align*}
s(\Vector{z}) \ =\ \Vector{z}^\top D Q D  \Vector{z}\ =\ \Vector{z}^\top D \left(\sum_{\substack{I,J \subseteq [n] \\ |I|,|J| \leq d}}c_{IJ}  \Vector{v}_{IJ} \Vector{v}_{IJ}^\top \right)  D  \Vector{z} \ =\ \sum_{\substack{I,J \subseteq [n] \\ |I|,|J| \leq d}}c_{IJ} \left( \Vector{z}^\top D \Vector{v}_{IJ} \right)^2,
\end{align*}
where $\struc{\Vector{z}}$ is the vector of $n$-variate monomials of maximal degree $d$, $\struc{Q}$ is a dd matrix, and $\struc{D}$ is a positive diagonal matrix. Since every vector $\Vector{v}_{IJ}$ has at most two nonzero entries, both equal to $\pm 1$, the SDSOS polynomial $s$ is always of the form $s(x)=\sum_{k} ( a_k p_k(x) + b_k q_k(x))^2$, where $p,q$ are monomials and $a_k, b_k \in \mathbb{R}$.

\subsubsection{SDSOS - Dual perspective}
\label{sec:SDSOS_dual}

For the BCPOP the dual of the problem~\eqref{eq:intro_PP_7} is a relaxation of the problem~\eqref{eq:intro_PP_5}. 
Indeed, similar as for formulation~\eqref{eq:intro_PP_5}, a conic duality theory can be used to transform program~\eqref{eq:intro_PP_7} into its dual of the form
\begin{align}
\label{eq:intro_PP_8}
\begin{aligned}
\struc{\BSDSOSD{f}{\mathcal{G}}{2d}} \ := \ \quad \min_{\mu:\{0,1\}^n\rightarrow \R}\{\tilde{\mathbb{E}}_\mu[f]~ |~ \tilde{\mathbb{E}}_\mu[1] = 1,~ \tilde{\mathbb{E}}_\mu[h] \geq 0,~ \text{for all }h \in \HierPut{\sSDSOS}{\constraints}{2d} \}  
\end{aligned}
\tag{$\Prog{\sSDSOSD}{\constraints}{2d}$}
\end{align}
for $\mathbb{\tilde{\mathbb{E}}}$ being a linear map, defined as in \cref{sec:sos_dual}. Analogously the \struc{$\overline{\text{SDSOS}}^{~*}$ \textit{program of degree} $2d$} takes the form $\struc{\ProgSch{\sSDSOSD}{\constraints}{2d}}$. 

Similar as in~\eqref{eq:pseudoexpectation_to_matrix}, Formulation~\eqref{eq:intro_PP_8} can be transformed into matrix form. In this case we obtain a set of $2 \times 2$ matrices that are required to be PSD. More formally, let $\struc{K} \subseteq \{I~|~ I\subseteq [n],~|I|\leq d\}$. For $M \in \R^{\binom{n}{\leq d} \times\binom{n}{\leq d}}$ being a real, symmetric matrix whose rows/columns are indexed	 with sets $I,J \subseteq [n]$ of size at most $d$, let $\struc{M_{\big|K}}$ be the principal submatrix of $M$ of entries that lie in the rows and columns indexed by the sets in $K$.

We obtain that $\Prog{\sSDSOSD}{\constraints}{2d}$ is equivalent to:
	 \begin{align}
\label{eq:intro_PP_8_2}
\begin{aligned}
\struc{\BSDSOSD{f}{\mathcal{G}}{2d}} \ = \ \quad \min_{\mu:\{0,1\}^n\rightarrow \R}\{\tilde{\mathbb{E}}_\mu[f]~ |~ \tilde{\mathbb{E}}_\mu[1] = 1,~ {M_g^{2d}}_{\big|I} \succeq 0,~ \text{for all } g \in \mathcal{G}, |I|=2 \}.
\end{aligned}
\end{align}

For BCPOP, both \eqref{eq:intro_PP_8} and~\eqref{eq:intro_PP_8_2} are solvable via an SOCP of size $\binom{n}{\leq d}^{O(1)}$. For more details we refer the reader to~\cite{AhmadiM17}.

\subsection{Sherali Adams}
\label{SubSec:PreliminariesSheraliAdams}

	An alternative method to approximate the sum of squares cone is based on nonnegative polynomials that depend on a limited number of variables, called \textit{\struc{$d$-juntas}}. The resulting program is called the \struc{\textit{Sherali Adams algorithm} (SA)} and was first introduced in~\cite{SheraliA90} as a method to tighten the linear programming relaxations for 0/1 hypercube optimization problems.
	Thus, we assume throughout the section and whenever we consider (SA) that the $\{0,1\}^n$ hypercube constraints are contained in $\mathcal{G}$, meaning that $\mathcal{G}=\{1,\pm(x_1^2-x_1),\ldots,\pm(x_n^2-x_n) , g_1,\ldots, g_m  \}$. 
	
	For $I \subseteq [n]$ we denote $\struc{\Vector{x}_I}=\prod_{i\in I}x_i$ and $\struc{\overline{\Vector{x}}}_I=\prod_{i \in I}(1-x_i)$. Let
	\begin{align*}
	\struc{\sSA} \ := \ \left\{ h \in \R[\Vector{x}]~\middle|~ h= \sum_{I,J} \alpha_{I,J} \cdot \Vector{x}_I \overline{\Vector{x}_J},~ \alpha_{I,J} \in \mathbb{R}_{\geq 0},~ I, J \subseteq [n] \right\}.
	\end{align*}
	A \struc{\emph{nonnegative $d$-junta}} is a function $f:\{0,1\}^n \to \mathbb{R}_{\geq 0}$ which depends only on at most $d$ input coordinates.
	It is easy to check that the set $\{h \in \R[\Vector{x}]~|~h \in \sSA,~\deg(h)\leq d \}$ is precisely the set of nonnegative $d$-juntas over the Boolean hypercube $\{0,1\}^n$. The \textit{\struc{degree-$2d$ Sherali Adams}}
	is the following problem:
	\begin{align}
	\label{eq:intro_PP_9}
	\begin{aligned}
	\struc{\BSA{f}{\constraints}{2d}} \ := \ \max\{\lambda \in \mathbb{R}~ |~ f-\lambda \in \HierPut{\sSA}{\constraints}{d}\};
	\end{aligned}
	\tag{$\Prog{\PSA}{\constraints}{2d}$}
	\end{align}
	analogously \struc{\textit{SA$^{*}$}} takes the form $\ProgSch{\PSA}{\constraints}{2d}$.
	Note that the superscript in $\HierPut{\sSA}{\constraints}{d}$ is $d$ (not $2d$), because of the way SA was defined historically, providing that $\BSA{f}{\constraints}{2d}\leq \BSOS{f}{\constraints}{2d}$. However, this does not affect the polynomial equivalence between the proof systems; see \cref{Def:PolynomialEquivalent}.
	
	
	The program $\PSA^{2d}$ can be solved using the \struc{\emph{linear program}} (LP) of size $\binom{n}{\leq d}^{O(1)}$.

    \subsubsection{SA - Dual perspective}
    \label{sec:SA_dual}
	
	Similarly as in \cref{sec:sos_dual} and \cref{sec:SDSOS_dual} one can use a conic duality theory to transform the program~\eqref{eq:intro_PP_9} into its dual of the form:
	 \begin{align}
\label{eq:intro_PP_9_2}
\begin{aligned}
\struc{\BSAD{f}{\mathcal{G}}{2d}} \ := \ \quad \min_{\mu:\{0,1\}^n\rightarrow \R}\left\{\tilde{\mathbb{E}}_\mu[f]~ |~ \tilde{\mathbb{E}}_\mu[1] = 1,~ \tilde{\mathbb{E}}_\mu[h] \geq 0,~ \text{for all }h \in \HierPut{\sSA}{\constraints}{2d}\right\}.  
\end{aligned}
\tag{$\Prog{\sSAD}{\constraints}{2d}$}
\end{align}	
	The program \eqref{eq:intro_PP_9_2} is a linear system of size $\binom{n}{\leq d}^{O(1)}$.
	Analogously, the \struc{$\overline{\text{SA}}^{~*}$ \textit{program of degree} $2d$} takes the form $\struc{\ProgSch{\sSAD}{\constraints}{2d}}$.

	\subsection{Sum of Nonnegative Circuit}
	A method for approximating the cone $\NNcone{\constraints}$, which is independent of $\TSOS$, is based on \struc{\textit{sums of nonnegative circuit polynomials} (SONC)}, defined below in \cref{sec:sonc_polynomials}. The technique was introduced by Iliman and the second author in~\cite{Iliman:deWolff:Circuits}.
	Let $\struc{\sSONC}$ be the \struc{\textit{set of finite sums of nonnegative circuit polynomials}}.
	We consider the following program:
	\begin{align}
	\label{eq:intro_PP_10}
	\begin{aligned}
	\struc{\BSONC{f}{\constraints}{2d}} \ = \ \max\{\lambda \in \mathbb{R}~ |~ f-\lambda \in \HierPut{\sSONC}{\constraints}{2d}\};
	\end{aligned}
	\tag{$\Prog{\PSONC}{\constraints}{2d}$}
	\end{align}
	analogously \struc{$\sSONC^{*}$} 
	takes the form $\ProgSch{\PSONC}{\constraints}{2d}$.
	As shown in~\cite[Theorem 4.8]{Dressler:Iliman:deWolff:Positivstellensatz}, for an arbitrary real polynomial that is strictly positive on a compact, basic closed semialgebraic set $\mathcal{G}$ there exists a $\TSONC^*$ certificate of nonnegativity, i.e., the Schm\"udgen-type Positivstellensatz \cref{thm:Schmudgen} applies to SONC. Moreover, searching through the space of degree $d$ certificates can be done via a \struc{\textit{relative entropy program (REP)}} \cite{Dressler:Iliman:deWolff:Positivstellensatz} of size $\binom{n+d}{d}^{O(1)}$; see also \cite{Chandrasekaran:Shah:RelativeEntropyApplications,Chandrasekaran:Shah:SignomialOptimization,Chandrasekaran:Murray:Wierman}. REPs are convex optimization programs and are efficiently solvable with interior point methods; see e.g., \cite{Chandrasekaran:Shah:RelativeEntropyApplications,nesterov} for more details.

\subsubsection{Nonnegative Circuit Polynomials}
\label{sec:sonc_polynomials}

We recall the most relevant statements about SONCs.

\begin{definition}
	A polynomial $f \in \R[\mathbf{x}]$ 
	is called a \struc{\emph{circuit polynomial}} if it is of the form
	\begin{align*}
	\struc{f(\mathbf{x})} ~:= ~ f_{\boldsymbol{\beta}} \mathbf{x}^{\boldsymbol{\beta}} + \sum_{j=0}^r f_{\boldsymbol{\alp}(j)} \mathbf{x}^{\boldsymbol{\alp}(j)}, 
	\end{align*}
	with $\struc{r} \leq n$, exponents $\struc{\boldsymbol{\alp}(j)}$, $\struc{\boldsymbol{\beta}} \in A$, and coefficients $\struc{f_{\boldsymbol{\alp}(j)}} \in \R_{> 0}$, $\struc{f_{\boldsymbol{\beta}}} \in \R$, such that $\New(f)$ is a simplex with even vertices $\boldsymbol{\alp}(0), \boldsymbol{\alp}(1),\ldots,\boldsymbol{\alp}(r)$ and the exponent $\boldsymbol{\beta}$ is in the strict interior of $\New(f)$.
	
	For every circuit polynomial we define the corresponding \struc{\textit{circuit number}} as
	\begin{align*}
	\struc{\Theta_f} \ := \ \prod_{j = 0}^r \left(\frac{f_{\boldsymbol{\alp}(j)}}{\lambda_j}\right)^{\lambda_j}. 
	\end{align*}
	\label{Def:CircuitPolynomial}
\end{definition}

One determines nonnegativity of circuit polynomials via its circuit number $\Theta_f$ as follows:

\begin{theorem}[\cite{Iliman:deWolff:Circuits}, Theorem 3.8]
	Let $f$  be a  circuit polynomial then $f$ is nonnegative if and only if $|f_{\boldsymbol{\beta}}| \leq \Theta_f$ and $\boldsymbol{\beta} \not \in (2\N)^n$ or $f_{\boldsymbol{\beta}} \geq -\Theta_f$ and $\boldsymbol{\beta }\in (2\N)^n$.
	\label{Thm:CircuitPolynomialNonnegativity}
\end{theorem}

Let
$$
\struc{\sSONC} \ := \ \left\{h= \sum_{i=1}^k \mu_i p_i~|~ 
p_i \text{ is a nonnegative circuit polynomial},~ \mu_i \geq 0,~ k \in \N^* \right\}.
$$

%

Following Reznick, we define maximal mediated sets; note that these objects are well-defined due to \cite[Theorem 2.2]{Reznick:AGI}

\begin{definition}
	Let $A \subseteq \Z^n$ such that $\vertices{\conv(A)} \in (2\Z)^n$. We call a set $M \subseteq \conv(A) \cap \Z^n$ \struc{($A$-)\textit{mediated}} if every element of $M$ is the midpoint of two distinct points in $\conv(A) \cap (2\Z)^n$.
	
	We define the \struc{\textit{maximal mediated set} $\conv(A)^*$} as the unique $A$-mediated set which contains every other $A$ mediated set.
	
	Let $\conv(A)$ be a simplex. If $\conv(A)^* = \conv(A) \cap \Z^n$, then we call $\conv(A)$ an \struc{$H$-\textit{simplex}}. If $\conv(A)^*$ consist only of $\vertices{\conv(A)}$ and the midpoints of the vertices, then we call $\conv(A)^*$ an \struc{$M$-\textit{simplex}}.
	\label{Definition:MaximalMediatedSets}
\end{definition}

Generalizing a result by Reznick in \cite{Reznick:AGI}, Iliman and the second author proved that maximal mediated sets are exactly the correct object for determining whether a nonnegative circuit polynomial is a sum of squares.

\begin{theorem}[\cite{Iliman:deWolff:Circuits}, Theorem 5.2]
	Let $f$ be a nonnegative circuit polynomial with inner term $f_{\Vector{\beta}} \Vector{x}^{\Vector{\beta}}$. Then $f$ is a sum of squares if and only if $f$ is a sum of monomial squares or if $\Vector{\beta} \in \New(f)^*$.
	
	Especially $f$ is always an SOS if $\New(f)$ is an $H$-simplex, and $f$ is never an SOS if $\New(f)$ is an $M$-simplex.
	\label{Theorem:NNCircuitPolyIsSOS}
\end{theorem}	
	
For further details about SONCs see e.g., \cite{deWolff:Circuits:OWR,Dressler:Iliman:deWolff:Positivstellensatz,DresslerKW18,Iliman:deWolff:Circuits,Seidler:deWolff:ExperimentalComparisonSONCandSOS}. A description of the dual of the SONC cone was recently provided in \cite{Dressler:Naumann:Theobald:DualSONCCone}, which we, however, do not need for the purpose of this article.
	
\subsection{Comparing proof systems}
\label{Sec:ComparingProofSystems}
In this section we introduce the notation used for comparing the proof systems presented in Section~\ref{sec:Preliminaries}, from the proof complexity perspective.
For a gentle introduction to proof complexity we refer the reader to e.g.,~\cite{Razborov16}.

Following the notation in~\cref{def:hierarchy_of_certificates}:
Let $\sGEN$ be a set of polynomials which we axiomatically assume to be nonnegative and $\constraints$ be a set of polynomials, which form the semialgebraic set $\feasible$. The \struc{GEN} \struc{\textit{proof system}} is the set of all algebraic derivations $f$ such that $f \in \bigcup_{d \in \N}  \HierPut{\sGEN}{\constraints}{2d}$ deducing nonnegativity of polynomials $f$ over $\feasible$. Analogously, the \struc{GEN$^{~*}$} \struc{\textit{proof system}} is the set of all algebraic derivations $f$ such that $f \in \bigcup_{d \in \N}  \HierSch{\sGEN}{\constraints}{2d}$ deducing nonnegativity of polynomials $f$ over $\feasible$.
The proof systems SOS, SOS$^{~*}$, SDSOS, SDSOS$^{~*}$, SA, SA$^{~*}$, SONC and SONC$^{~*}$ are defined analogously.

The complexity of the certificate depends on the $d$ needed to certify the nonnegativity. Revising~\cref{Def:PolynomialEquivalent} we say that a proof system $P$ \textit{\struc{contains}} a proof system $Q$ if for every set of polynomials $\constraints$ and a polynomial $f$ admitting a degree $d$ certificate of nonnegativity over $\feasible$ in $Q$, $f$ admits also a degree $O(d)$ certificate in $P$. A system $P$ \textit{\struc{strictly contains}} $Q$ if $P$ contains $Q$ but $Q$ does not contain $P$. I.e., there exist at least one set $\constraints$ and a polynomial $f$ nonnegative over $\feasible$ such that $f$ admits a degree $d$ certificate in $P$ but for every $c \in \N$ $f$ does not admit a degree $cd$ certificate in $Q$.
Systems $P$ and $Q$ are \textit{\struc{polynomially equivalent}} if $P$ contains $Q$ and $Q$ contains $P$. Finally, systems $P$ and $Q$ are \textit{\struc{polynomially incomparable}} if neither $P$ contains $Q$, nor $Q$ contains $P$. I.e., there exist sets $\constraints$, $\constraints^{'}$ and polynomials $f$, $f^{'}$ nonnegative over $\feasible$, $\feasible^{'}$, respectively, such that $f$ admits a degree $d$ certificate in $P$ but for every $c \in \N$ $f$ does not admit a degree $cd$ certificate in $Q$ and $f^{'}$ admits a degree $d^{'}$ certificate in $Q$ but for every $c \in \N$ $f^{'}$ does not admit a degree $cd^{'}$ certificate in $P$.

	\section{SOS vs. SONC}
	\label{Sec:SONCvsSOS}
	
	
	It is well-known that the $\sSOS$ cone and $\sSONC$ cone are not contained in each other~\cite[Proposition 7.2]{Iliman:deWolff:Circuits} 
	This statement, however, gives no prediction whether or not for CPOPs these systems are polynomially equivalent or not.
	 In this section we show that for every $n$ there exist CPOPs such that the difference between the minimal degrees of a SOS and a SONC certificate is arbitrarily large and vice versa. 
	 
	\subsection{SONC does not contain SOS}
	\label{sec:SOS_better_than_SONC}
	
	We consider the following family of polynomials:
	
	\begin{definition}
		We define the family of \struc{\textit{signed quadrics} $(N_n)_{n \in \N^*}$} by $\struc{N_n} := \left(1 - \sum_{j = 1}^n x_j \right)^2$.
	\end{definition}
	
	It is obvious that every $N_n$ is SOS and that its zero set is the unit ball of the 1-norm, i.e., for all $n \in \N^*$ we have
	\begin{align}
\cV(N_n) \ = \ \{\mathbf{x} \in \R^n \ : \ ||\mathbf{x}||_1 = 1\}.
\label{equ:ZeroSetSquared1NormPolynomial}
\end{align}
The support of $N_2$ and $N_3$ is depicted together with their Newton polytopes in \cref{fig:SignedQuadric}.
\begin{figure}
	\ifpictures
	\includegraphics[width=0.25\linewidth]{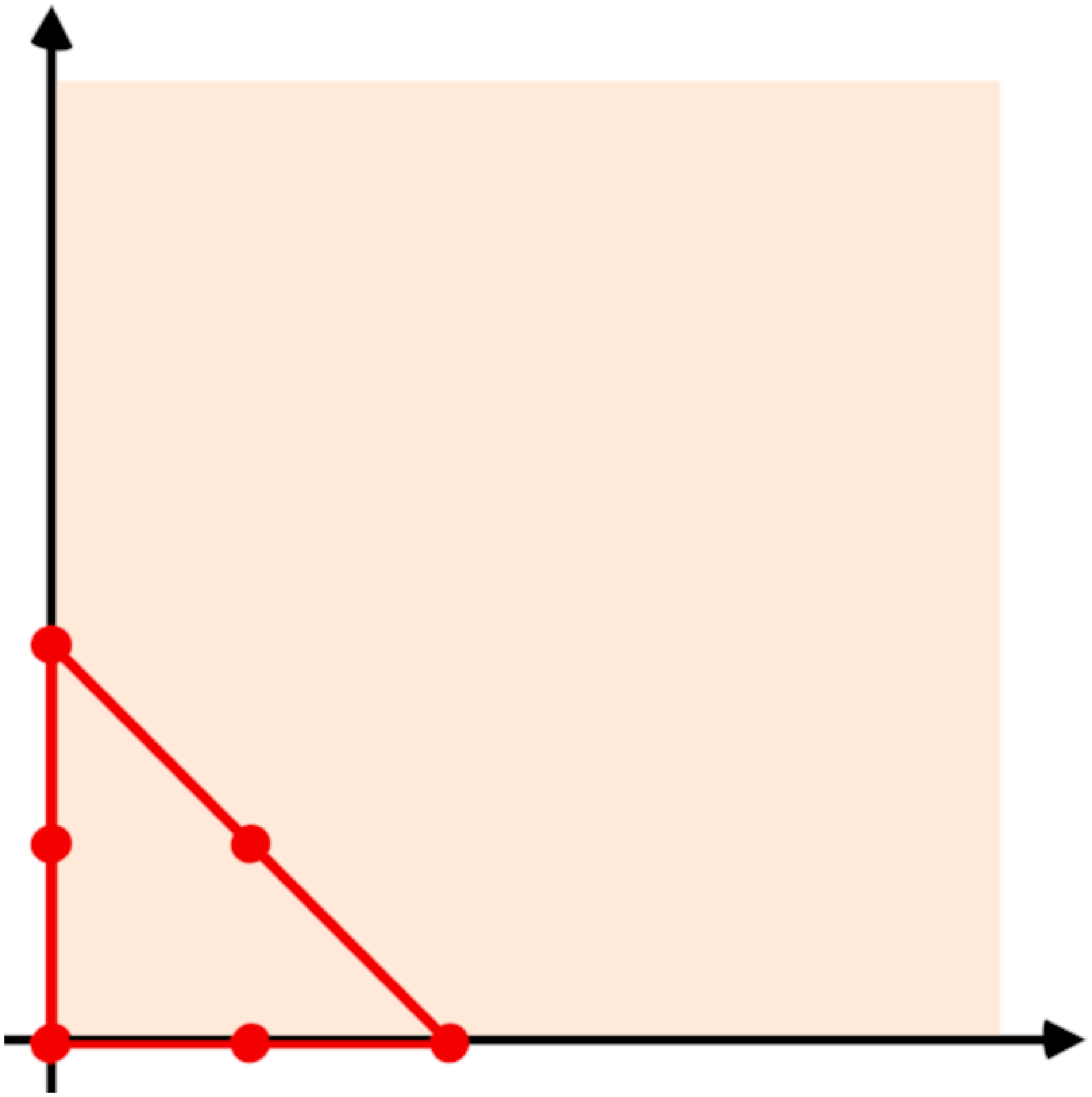} \qquad \qquad
	\includegraphics[width=0.35\linewidth]{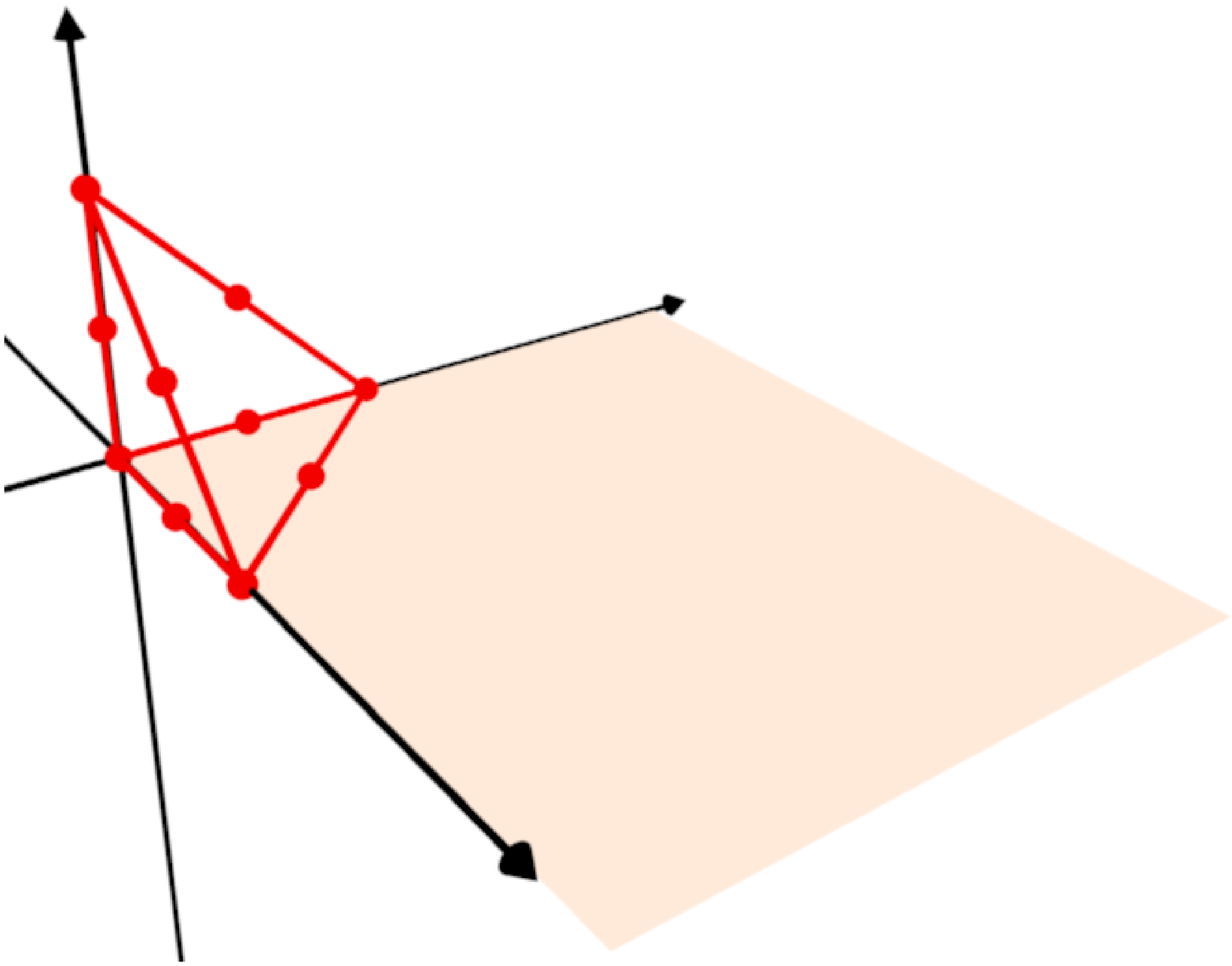}
	\fi
	\caption{The support and the Newton polytopes of $N_2$ and $N_3$.}
	\label{fig:SignedQuadric}
\end{figure}
	It is known that for every $n$ the function $N_n$ cannot be written as a combination of $(n-1)$-juntas \cite[Theorem 1.12]{Lee15}. It is, however, also straightforward to conclude that for every $n \in \N$ the polynomial $N_n$ is not a SONC polynomial.

	\begin{lemma}
		For all $n \geq 2$ it holds that $N_n \notin \sSONC$.
		\label{lemma:Squared1NormPolynomialIsNotSONC}
	\end{lemma}
	
	\begin{proof}
		By~\cref{equ:ZeroSetSquared1NormPolynomial} the real zero set $\cV(N_n)$ is equal to the boundary of the $n$-dimensional cross-polytope; see e.g., \cite{Ziegler:Book}. In particular, it is an $(n-1)$ dimensional piecewise-linear set. A SONC, however, has at most $2^n$ many distinct real zeros by \cite[Corollary 3.9]{Iliman:deWolff:Circuits}.
	\end{proof}
	

		In~\cite[Example 3.7]{Seidler:deWolff:ExperimentalComparisonSONCandSOS} it is shown that $N_2$ is not a SONC due to a term by term inspection. We point out that one could build over that argument and reprove inductively~\cref{lemma:Squared1NormPolynomialIsNotSONC} using the fact that the support set of $N_2$ equals the restriction of the support set of $N_n$ restricted to a specific $2$-face of $\New(N_n)$.
	
	\begin{corollary}
		For every $n \in \N$ with $n \geq 2$ and every $t \geq 1$ there exist infinitely many systems $\mathcal{G}$ such that $\HierPut{\sSOS}{\constraints}{2} \not\subset \HierSch{\sSONC}{\constraints}{2t}$.
		\label{corollary:SOSCanBeArbitrarilyBetterThanSONC}
	\end{corollary}
	
	\begin{proof}
		Let $n \in \N_{\geq 2}$ be fixed. Consider a system 
		\begin{align}
		\min N_n \ \text{ such that } \ g_1,\ldots,g_s \geq 0
		\label{equ:SOSFriendlyCPOP}
		\end{align}
		where $N_n$ is the signed quadric and $\constraints = \{g_1,\ldots,g_s\}$ is a system of polynomials such that $\min_{i \in [s]}\{\deg(g_i)\} \geq 2t+1$, $\cV(N_n) \subseteq \feasible$, and $\feasible$ compact.
		
		On the one hand, there exists an SOS certificate of degree $2$ for the system \cref{equ:SOSFriendlyCPOP} given by $N_n$ alone, as $N_n$ is already an SOS and moreover $\min_{\mathbf{x} \in \R^n} N_n = \min_{\mathbf{x} \in \feasible} N_n$.
		
		On the other hand, due to the Positivstellensatz result for SONCs~\cite[Theorem 4.8]{Dressler:Iliman:deWolff:Positivstellensatz}, there exists a SONC certificate of the form $N_n \in \HierSch{\sSONC}{\constraints}{2d}$, for some value of $d$. But since $N_n$ is not a SONC due to \cref{lemma:Squared1NormPolynomialIsNotSONC} the certificate necessarily has to incorporate at least one of the constraints defining the set $\constraints$. Hence, $d>t$.
	\end{proof}

	\subsection{SOS does not contain SONC}
	\label{sec:SONC_better_than_SOS}
	
	In this section we show the inverse of the result from \cref{sec:SOS_better_than_SONC}, namely that SONC is not contained in SOS.
	
	\begin{definition}[Generalized Motzkin Polynomial]
		Let $\struc{\mathbf{e}} := \sum_{j = 1}^{n} \mathbf{e}_j$. For every $n \in \N$ with $n \geq 2$ we define the \struc{\textit{Generalized Motzkin Polynomial}} as
		\begin{align*}
		\struc{M_{n}} \ := \ 1 + \sum_{j=1}^n \mathbf{x}^{2(\mathbf{e} + \mathbf{e}_j)} - (n+1)\mathbf{x}^{\mathbf{e}}.
		\end{align*}
	\end{definition}
	
	Note the $M_2$ is the usual \textit{\struc{Motzkin polynomial}} $1 + x_1^2x_2^4 + x_1^4x_2^2 - 3x_1^2x_2^2$. The support of $M_2$ and $M_3$ is depicted together with their Newton polytopes in \cref{fig:GeneralizedMotzkin}.
	\begin{figure}
	\ifpictures
	\includegraphics[width=0.25\linewidth]{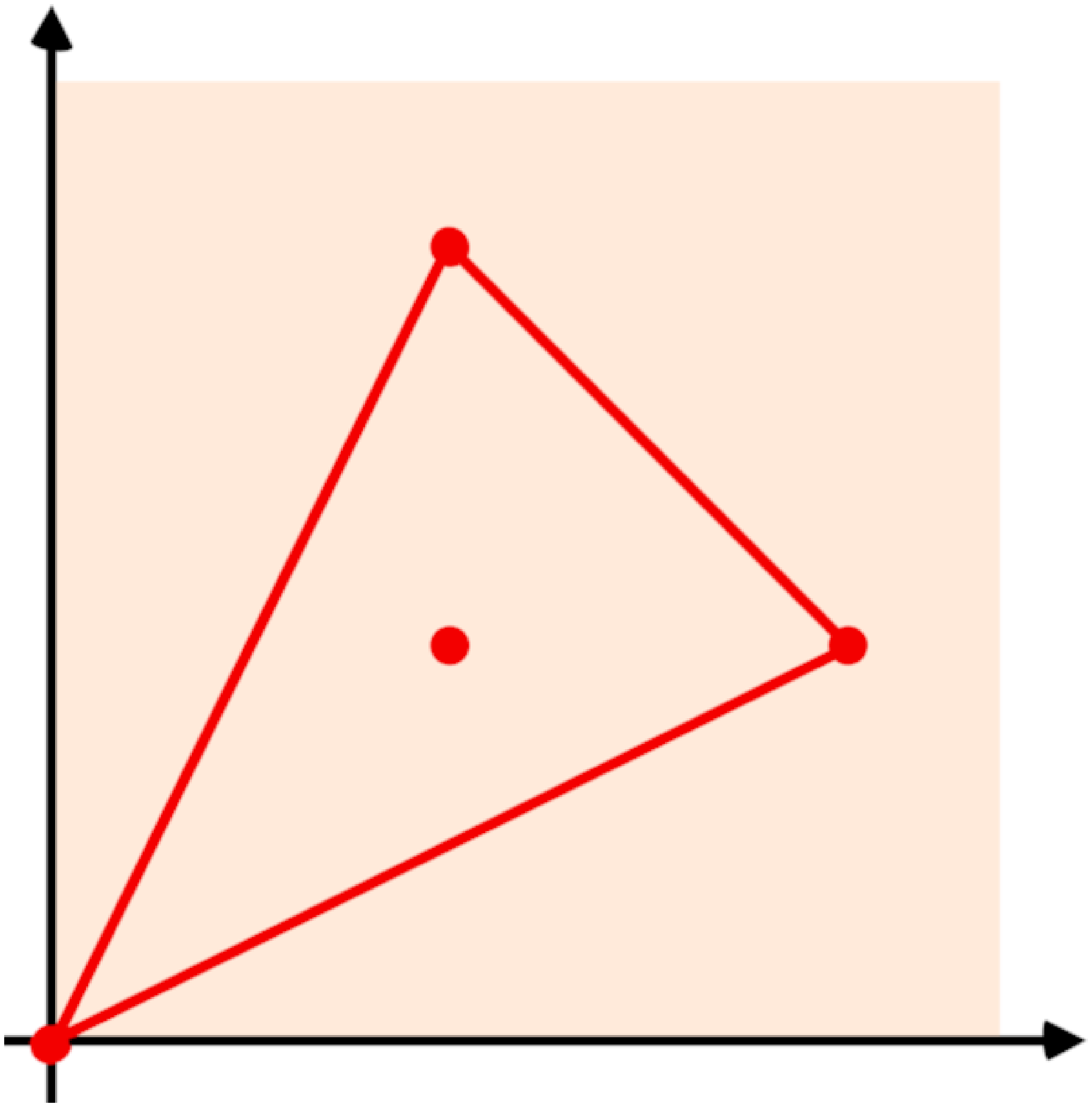} \qquad \qquad
	\includegraphics[width=0.35\linewidth]{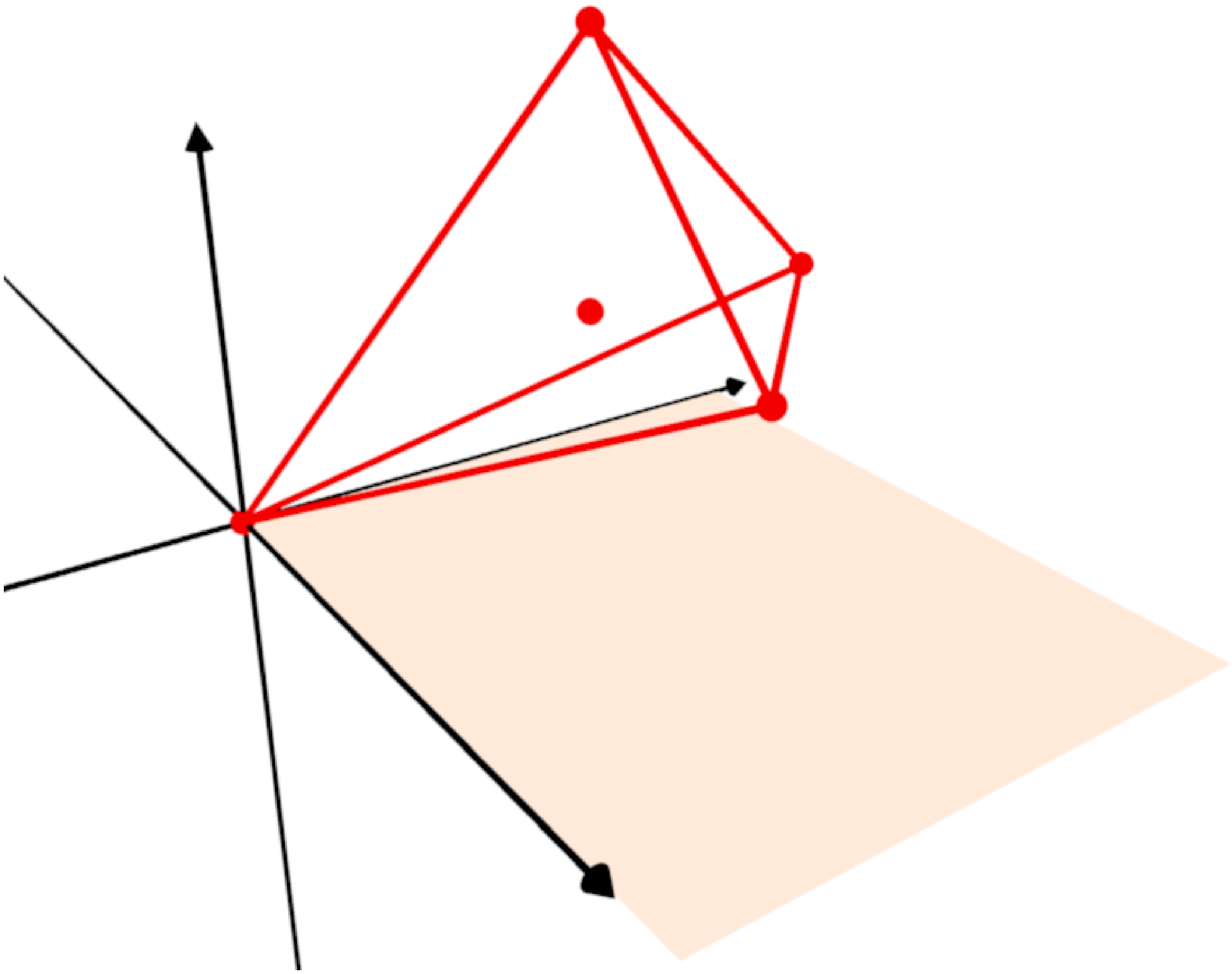}
	\fi
	\caption{The support and the Newton polytopes of $M_2$ and $M_3$.}
	\label{fig:GeneralizedMotzkin}
\end{figure}

	\begin{proposition}
		For every $n \in \N_{\geq 2}$ we have $M_n \in \sSONC$, but $M_n \notin \sSOS$, and moreover $\mathcal{V}(M_n) = \{-1,1\}^n$.
		\label{prop:GeneralizedMotzkin}
	\end{proposition}
	
	Note that $M_n$ not being $\sSOS$ for $n \geq 2$ was in fact already shown by Motzkin (not only $n = 2$); see \cite{Motzkin67} and see also \cite[Section 6]{Reznick:AGI}. Furthermore, $\New(M_n)^*$ being an $M$-simplex, which implies $M_n$ not being $\sSOS$ was shown by Reznick in \cite[Theorem 6.9]{Reznick:AGI}. We provide an own proof of all these facts here for convenience of the reader.
	
	
	\begin{proof}
		Let $n \geq 2$. First, we show that $M_n$ is a nonnegative circuit polynomial which vanishes exactly on the Boolean hypercube $\{-1,1\}^n$.
		
		According to \cref{Def:CircuitPolynomial} $M_n$ is a circuit polynomial with inner term $-(n+1)\mathbf{x}^{2\mathbf{e}}$. A straightforward computation yields for the barycentric coordinates $\lam_j = 1/(n+1)$ for every $j= 0,\ldots,n$ and the circuit number satisfies $\Theta_{M_n} = n+1$. Thus, $M_n$ is nonnegative by \cref{Thm:CircuitPolynomialNonnegativity}. As mentioned before, every nonnegative circuit polynomial has at most one zero on every orthant \cite[Corollary 3.9]{Iliman:deWolff:Circuits}. Moreover, an evaluation shows that $M_n(\mathbf{x}) = 0$ for every $\mathbf{x} \in \{-1,1\}^n$. Thus, $\mathcal{V}(M_n) = \{-1,1\}^n$.
		
		Second, we show that $M_n$ is not a sum of squares. Since $M_n$ is a nonnegative circuit polynomial, which is not a sum of monomial squares, this is equivalent to the fact that the lattice point $2\mathbf{e}$ does not belong to the maximal mediated set $\New(M_n)^*$ by \cref{Theorem:NNCircuitPolyIsSOS}. Here, we show more generally that $\New(M_n)$ is even an $M$-simplex, which implies $2\mathbf{e} \notin \New(M_n)^*$. We observe:
		\begin{align}
		\New(M_n) \cap (2\Z)^n = \vertices{\New(M_n)} \cup \{2\mathbf{e}\}. \label{equ:ProofGeneralizedMotzkin1}
		\end{align}
		This follows from 
		\begin{align}
		\sum_{\mathbf{v} \in \vertices{\New(M_n))}} \frac{1}{n+1} \mathbf{v} = 2\mathbf{e}
		\label{equ:ProofGeneralizedMotzkin2}
		\end{align}
		and the fact that every lattice point $\mathbf{w} \in (2\Z^n)$ with $||\mathbf{w} - 2\mathbf{e}||_1 = 2$ satisfies either $\mathbf{w} \in \vertices{\New(M_n)}$, or $\mathbf{w} \notin \New(M_n)$ due to $w_i = 0$, $w_j \neq 0$ for some $i,j \in [n]$ with $i \neq j$.
		
		We have $2\mathbf{e} \notin \New(M_n)^*$. By definition, every point in $\New(M_n)^*$ is the midpoint of two distinct points in $\New(M_n)^* \cap (2\Z)^n$. This is impossible due to \eqref{equ:ProofGeneralizedMotzkin1}, \eqref{equ:ProofGeneralizedMotzkin2}, and the fact that the convex combination in \eqref{equ:ProofGeneralizedMotzkin2} is unique since $\New(M_n)$ is a simplex.
		Thus, $\New(M_n)^* \cap (2\Z)^n = \vertices{M_n}$ and $\New(M_n)$ is an $M$-simplex by \cite[Theorem 2.5]{Reznick:AGI}.
	\end{proof}

	\begin{corollary}
		For every $n \in \N$ with $n \geq 2$ and every $t \geq n+1$ there exist infinitely many systems $\mathcal{G}$ such that $ \HierPut{\sSONC}{\constraints}{n+1} \not\subset \HierSch{\sSOS}{\constraints}{t}$. 
		\label{corollary:SONCCanBeArbitrarilyBetterThanSOS}
	\end{corollary}
	
	\begin{proof}
		Let $n \in \N_{\geq 2}$ be fixed. Consider a system 
		\begin{align}
		\min M_n \ \text{ such that } \ g_1,\ldots,g_s \geq 0
		\label{equ:SONCFriendlyCPOP}
		\end{align}
		where $M_n$ is the generalized Motzkin polynomial and $\constraints = \{g_1,\ldots,g_s\}$ is a system of polynomials such that $\min_{i \in [s]}\{\deg(g_i)\} \geq t+1$, $\{\pm 1\}^n \subseteq \feasible$ and $\feasible$ compact.
		
		On the one hand, there exists a SONC certificate of degree $n+1$ for the system  \cref{equ:SONCFriendlyCPOP} given by $M_n$ alone as $M_n$ is already a SONC by \cref{prop:GeneralizedMotzkin} and moreover $\min_{\mathbf{x} \in \R^n} M_n = \min_{\mathbf{x} \in \feasible} M_n$.
		
		On the other hand, due to \cref{thm:intro_Putinar}, there exists an SOS certificate of the form $M_n \in \HierSch{\sSOS}{\constraints}{d}$ for some $d \in \mathbb{N}_+$. But since $M_n$ is not a SOS due to \cref{prop:GeneralizedMotzkin}, the certificate has to necessarily involve at least one constraint from $\constraints$. Thus $d>t$.
	\end{proof}
	
	\begin{corollary}
		\label{cor:sonc_vs_sos}
		The pairs of systems $SONC$ and $SOS$; $SONC$ and $SOS^{~*}$; $SONC^{~*}$ and $SOS$; $SONC^{~*}$ and $SOS^{~*}$ are polynomially incomparable.
	\end{corollary}
	
	\begin{proof}
		Follows immediately from \cref{corollary:SOSCanBeArbitrarilyBetterThanSONC} and \cref{corollary:SONCCanBeArbitrarilyBetterThanSOS}.
	\end{proof}

	
	\section{SDSOS vs SONC}
	\label{Sec:SDSOSvsSONC}

	In this section, we show that for constrained polynomial optimization problems CPOPs, SONC strictly contains SDSOS. The same relations holds for the SONC$^*$ and for SDSOS$^*$ algorithms.
	
	\subsection{SDSOS is SONC}
	We start with proving that every SDSOS polynomial of degree $d$ is also a SONC polynomial of degree $d$.
	
	\begin{lemma}
		\label{lem:f_sdoss_is_sonc}
		Every scaled diagonally dominant polynomial is a circuit polynomial.
	\end{lemma}	
	\begin{proof}
		By~\cite{AhmadiM14} we know that every scaled diagonally dominant polynomial $s$ can be written as a sum of binomial squares, i.e., of the form $s(\Vector{x})=\sum_{k} s_k(\Vector{x})$, for $s_k(\Vector{x})= ( a_k p_k(\Vector{x}) + b_k q_k(\Vector{x}))^2$, where $p,q$ are monomials and $a_k,~b_k \in \mathbb{R}$. Thus, $s(\Vector{x})=\sum_{k} ( a_k p_k(\Vector{x}))^2 + (b_k q_k(\Vector{x}))^2 -2 a_k p_k(\Vector{x})b_k q_k(\Vector{x}) $, where $a_k^2,~b_k^2 \in \mathbb{R}_{\geq 0}$ for every index $k$ in the summation. Moreover, by \cref{Def:CircuitPolynomial}, for every $k$, $\New(s_k)$ is a one dimensional simplex with two even vertices $\boldsymbol{\alp}(0), \boldsymbol{\alp}(1)$, given by the exponents of the squared monomials, and the exponent $\boldsymbol{\beta}$ of the term $-2 a_k p_k(\Vector{x})b_k q_k(\Vector{x})$ is in the strict interior of $\New(f)$ since $\boldsymbol{\beta}=1/2~\boldsymbol{\alp}(0)+1/2~\boldsymbol{\alp}(1)$. Finally, the circuit number $\Theta_{s_k}$ is equal to $\left(\frac{f_{\boldsymbol{\alp}(0)}}{\lambda_0}\right)^{\lambda_0}\cdot \left(\frac{f_{\boldsymbol{\alp}(1)}}{\lambda_1}\right)^{\lambda_1} = 2a_kb_k$, thus by \cref{Thm:CircuitPolynomialNonnegativity}, $s_k$ is a nonnegative circuit polynomial.
	\end{proof}
	
	We note that a similar statement to~\cref{lem:f_sdoss_is_sonc} was very recently, independently observed in \cite{Chandrasekaran:Murray:Wierman}.
	
	\begin{corollary}
		For every $d \in \N^*$ we have $\BSDSOS{f}{\constraints}{2d} \leq \BSONC{f}{\constraints}{2d}$ and $\BSDSOS{f}{\preprime{\constraints} }{2d} \leq \BSONC{f}{\preprime{\constraints}}{2d}$.
		\label{cor:SONCcontainsSDSOS}
	\end{corollary}
	\begin{proof}
		Assume that for some polynomial $f \in \mathbb{R}[\Vector{x}]$ and $\lambda \in \mathbb{R}$ there exists a $d \in \N^*$ such that $f-\lambda \in \HierPut{\sSDSOS}{\constraints}{2d}$. We show that necessarily it has to hold $f-\lambda \in \HierPut{\sSONC}{\constraints}{2d}$.		
		Since $f -\lambda \in \HierPut{\sSDSOS}{\constraints}{2d}$ there necessarily exist SDSOS polynomials $s_i$ such that $f-\lambda=\sum_{i=0}^m  s_i g_i$ for $g_i \in \mathcal{G}$ and $\deg(s_i g_i) \leq 2d$ for every $i \in [m]$. Moreover, by \cref{lem:f_sdoss_is_sonc} every SDSOS polynomial is a SONC polynomial, thus we have $f-\lambda=\sum_{i=0}^m  s_i g_i \in \HierPut{\sSONC}{\constraints}{2d}$. The proof works analogously for the second inequality.
	\end{proof}
	
	It remains to show that the for every $n$ there exist CPOPs such that the ratio of the minimal degrees of a SDSOS and a SONC is not bounded by a constant. 
	

	\begin{corollary}
		\label{cor:sdsos_vs_sonc}
		SONC proof system strictly contains SDSOS proof system and the same relation holds for SONC$^{~*}$ and SDSOS$^{~*}$. 
	\end{corollary}

	\begin{proof}
		The corollary follows from \cref{cor:SONCcontainsSDSOS} and the fact that $\HierPut{\sSDSOS}{\constraints}{2d} \subseteq \HierPut{\sSOS}{\constraints}{2d}$ ($\HierSch{\sSDSOS}{\constraints}{2d} \subseteq \HierSch{\sSOS}{\constraints}{2d}$) together with \cref{prop:GeneralizedMotzkin}.
	\end{proof}

%
%
	
	As a consequence we show that there exists no equivalent of Putinar's Positivstellensatz for SDSOS algorithm. 
	
	\begin{corollary}
		There exist $f \in \R[\Vector{x}]$ and infinitely many systems of polynomials $\mathcal{G} = \{g_1,\ldots,g_s\} \subset \R[\Vector{x}]$ such that $\mathcal{G}_+$ is Archimedean, $f(\Vector{x}) > 0$ for all $\Vector{x} \in \mathcal{G}_+$ but for all $d \in \N$, $f \notin \HierPut{\sSDSOS}{\constraints}{2d}$.
		\label{Corollary:ThereExistsNoPutinarPositivstellensatzForSDSOS}
	\end{corollary}

    \begin{proof}
    	By constructing an explicit example, it was shown in~\cite[Theorem 5.1]{DresslerKW18} that the equivalent statement holds for the SONC case. Thus, the statement follows immediately from \cref{cor:SONCcontainsSDSOS}.
    \end{proof}

\subsection{Closures under Changes of Bases and Relations to SOCP}

In the rest of this section we provide two observations regarding the behaviour of $\sSONC$ and $\sSDSOS$ under a change of bases and their relation to second order cone programming.\\

	In \cite[Lemma 4.1]{Dressler:Iliman:deWolff:Positivstellensatz} the authors showed that the SONC cone is not closed under multiplication, i.e., if $s_1,s_2 \in \PSONC^{2d}$, then this does not imply $s_1 \cdot s_2 \in \PSONC^{4d}$ in general. Moreover, $\PSONC^{2d}$ is not closed under affine transformations or more generally a change of bases; see \cite[Corollary 3.2]{DresslerKW18}. 
	These results are in sharp contrast to the SOS cone, which is closed both under multiplication and under a change of bases. 
	Similarly as for SONC, it is well-known that $\sSDSOS$ is not closed under multiplication, and a change of bases; see e.g., \cite{Ahmadi:Hall}.
	
	More precisely, we define 
	\begin{align*}
		\struc{\closSDSOS} \ := \ \left\{f(y_1,\ldots,y_n) \in \sSDSOS \ : \ 
		\begin{array}{l}
			\text{there exists a } T \in \text{GL}_n \text{ with } \\
			y_i = x_i^{T\Vector{e}_i} \text{ for all } i \in [n]
		\end{array}
		\right\};
	\end{align*}
    analogously for $\struc{\closSONC}$.
	
	Currently, it is an open problem in the community to decide whether the closure of $\PSDSOS^{2d}$ with respect to a change of bases equals $\PSOS^{2d}$. 
	
	\begin{problem}
		Is $\closSDSOS = \sSOS$?
	\end{problem}

	From the results of this section we obtain the following corollary.
	\begin{corollary}
        If $\closSDSOS = \sSOS$, then $\sSOS \subsetneq \closSONC$.
    \label{Corollary:Closures}
	\end{corollary}
	
	\begin{proof}
		Follows directly from \cref{lem:f_sdoss_is_sonc}, which does not depend on the chosen basis and $\sSOS \neq \sSONC$; see \cite[Proposition 7.2]{Iliman:deWolff:Circuits}.
	\end{proof}

	Moreover we obtain the following consequence about the relation of $\sSONC$ and \struc{\textit{second order cone programming} $\sSOCP$}. Since we use $\sSOCP$'s only in the following corollary, we omit a full definition of SOCP and refer the reader to the standard literature like \cite{Boyd04}. The bounds $\struc{\BSOCP{f}{\constraints}{2d}}$ and $\struc{\BSOCP{f}{\preprime{\constraints}}{2d}}$ are defined analogously to the other hierarchies.
    
    \begin{corollary}
    	For every $d \in \N^*$ we have $\BSOCP{f}{\constraints}{2d} \leq \BSONC{f}{\constraints}{2d}$ and $\BSOCP{f}{\preprime{\constraints}}{2d} \leq \BSONC{f}{\preprime{\constraints}}{2d}$.
		\label{cor:SONCcontainsSOCP}
    \end{corollary}
    
    \begin{proof}
    	It was shown by Ahmadi and Majumdar that every $\sSDSOS$ certificate is $\sSOCP$ in \cite[Theorem 10]{AhmadiM17}, and, very recently, that every $\sSOCP$ certificate is $\sSDSOS$ by Ding and Lim \cite[Theorem 3.3]{Ding:Lim:HigherOrderConeProgramming}. Thus, the statement follows immediately from \cref{cor:SONCcontainsSDSOS}.
    \end{proof}
    
    We remark that, however, Averkov \cite[Theorem 17]{Averkov:OptimalSizeLMI} recently showed that the \textit{semidefinite extension degree} of SONC equals two, and thus SONC is an \struc{\textit{SOCP-lift}}; see also \cite{Gouveia:Parrilo:Thomas} for further details on lifts.
	
	\section{Hierarchies on the Boolean Hypercube}

    In this section we prove the dependencies between various hierarchies on the Boolean hypercube $\{0,1\}^n$.
	
Let, for this section, $\constraints$ be a collection of polynomials such that for all $i \in [n]$ we have $\pm(x_i^2 - x_i) \in \constraints$ and $l_i := N \pm x_i \in \constraints$ for $N \in \R_{> 1}$, such that $\feasible \subseteq\{0,1\}^n$.
Let $\sGEN \subset \NNcone{\feasible}$ be an arbitrary class of polynomials, which are nonnegative on the Boolean hypercube. We consider the corresponding optimization problem~\eqref{Def:PolynomialHierarchy}.
%
	We start with proving a general statement saying that every proof certificate that can certify nonnegativity of an $n$-variate polynomial over the unconstrained Boolean hypercube with an $O(n)$ degree certificate is at least as strong as the $\TSA^*$ hierarchy.
	
	\begin{theorem}
		\label{thm:General_proof_as_strong_as_SA}
		Let $f \in \R[\Vector{x}]$. Assume that there exists a $c \in \N^*$  such that for every $n \in \N$, and $\constraints:=\{ \pm(x_i^2 - x_i):~ i \in [n] \}$ such that
		\begin{align}
			\HierSch{\sGEN}{\constraints}{2cn} \ = \ \NNcone{\{0,1\}^n}.
			\label{Equ:ConditionGeneralAsStrongAsSA}
		\end{align}
        Then for every finite set of polynomial constraints $\mathcal{G'}$ with $\mathcal{G}'_+ \subseteq \{0,1\}^n$ and for every $d\in \mathbb{N}$ with $~d\leq n$ it holds that $\BoundSch{f}{\PSA}{\constraints'}{2d} \leq \ \BoundSch{f}{\PGEN}{\constraints'}{2cd}$.
	\end{theorem}
	
	\begin{proof}
		Consider a polynomial $f \in \mathbb{R}[\Vector{x}]$, a set $\mathcal{G^{'}_+} \subseteq \{0,1\}^n$ and a real number $\lambda$ such that $f-\lambda \in \HierSch{\sSA}{\mathcal{G^{'}}}{2d}$. We show that $f-\lambda \in \HierSch{\sGEN}{\constraints'}{2cn}$. 
		Let $\struc{\SizePrep{2d}^{'}}$ be the cardinality of the set $\preprime{\constraints^{'}}$ restricted to polynomials of degree at most $2d$.
		By definition, $f-\lambda \in \HierSch{\sSA}{\mathcal{G^{'}}}{2d}$ implies that there exists a certificate $f-\lambda=\sum_{i=0}^{\SizePrep{2d}^{'}} p_i G_i$, for $G_i \in \preprime{\mathcal{G}^{'}}$ such that every polynomial $p_i$ is of the form $p_i = \sum_{\ell = 0}^{r_i} q_{i\ell}$ such that every $q_{i\ell}$ is a nonnegative $k_i$-junta with $\struc{k_i}:=\lfloor \frac{2d-\deg(G_i)}{2} \rfloor$. 
		
		For every $i\in [\SizePrep{2d}^{'} ]$ we consider a set of polynomials $\constraints^{''}_i \subseteq \constraints'$ such that $\{\pm(x_j^2 - x_j), l_j \ | \ j \in [k_i]\} \subseteq \constraints^{''}_i$ and $\mathcal{G}^{''}_{i+} = \{0,1\}^{k_i}$.
         Let $\struc{\SizePrep{i,2d}^{''}}$ be the cardinality of the set $\preprime{\constraints^{''}_i}$ restricted to polynomials of degree at most $2d$.
         By the assumption \eqref{Equ:ConditionGeneralAsStrongAsSA} there exists a $c\in\mathbb{N}^*$ such that every $k$-variate polynomial, which is nonnegative over the Boolean hypercube $\{0,1\}^k$, has a degree $2ck$ certificate using polynomials in $\sGEN$. Thus, we have in particular $q_{i\ell} \in \HierSch{\sGEN}{\mathcal{G}^{''}_i}{2ck_i}$ for every $i \in [ \SizePrep{2d}^{'} ]$ and $\ell \in [r_i]$. Hence, we can write $q_{i\ell}=\sum_{j=0}^{ \SizePrep{2d}^{''}  }s_{i\ell j} G_{i\ell j}$ such that  $ s_{i \ell j} \in \sGEN$, $G_{i \ell j} \in \preprime{\mathcal{G}^{''}_i}$, and for every $j\in [ \SizePrep{2d}^{''} ]$ we have $\deg(s_{i \ell j}G_{i \ell j}) \leq 2ck_i $. In summary, we obtain:
		\begin{align*}
			f-\lambda \ = \ \sum_{i=0}^{ \SizePrep{2d}^{'}  } \sum_{\ell = 0}^{r_i} \left(\sum_{j=0}^{ \SizePrep{2d}^{''}  }s_{i \ell j} G_{i \ell j}\right) G_i,
		\end{align*}
        where for every $i \in [  \SizePrep{2d}^{'}  ]$, $j \in [ \SizePrep{2d}^{''} ]$, and $\ell \in [r_i]$ the degree $\deg(s_{i \ell j}G_{i \ell j}G_i)$ is at most $2ck_i+\deg(G_i)=2c \lfloor \frac{2d-\deg(G_i)}{2} \rfloor +\deg(G_i)\leq 2cd$, and hence 
        $$f-\lambda \in \HierSch{\sGEN}{\constraints' \cup \bigcup_{i=0}^{ \SizePrep{2d}^{''} } \constraints^{''}_i}{2cd}.$$
        By the containment $\mathcal{G}^{''}_i \subseteq \mathcal{G}^{'}$, for every $i \in [ \SizePrep{2d}^{''} ]$, we get that 
        $$
        \constraints' \cup \bigcup_{i=0}^{ \SizePrep{2d}^{''} } \constraints^{''}_i = \constraints^{'},
        $$  and  the statement follows.
	\end{proof}
	\subsection{Properties of the $\sGEN$ proof system}
	
	In \cref{thm:General_proof_as_strong_as_SA} we gave a sufficient condition for the proof system to be at least as strong as SA$^*$. As a consequence every proof system satisfying that condition attains the properties of the SA proof system. In particular, this applies to the \struc{\emph{conditioning}} property, which has been widely used to construct algorithmic results for various BCPOP problem, see e.g.,~\cite{LeveyR16}. In what follows we provide a formal description of the property.
	
	\begin{lemma}[Conditioning]
		\label{lem:SA_conditioning}
		For every $d \in \mathbb{N}^*$, let $\struc{\mathbb{L}}$ be the linear operator feasible for $\Prog{\sSAD}{\constraints}{2d}$.
        Let $i\in [n]$ be an index such that $0<\mathbb{L}[x_i]<1$. We define $ \struc{\mathbb{L}_{i,(0)}},\struc{\mathbb{L}_{i,(1)}}~:~ \mathbb{R}[\Vector{x}] \to \mathbb{R}$ such that:
        {\small
		\begin{align*}
			\mathbb{L}_{i,(1)}\left[\prod_{j \in J} x_j\right]~:=~\frac{\mathbb{L}\left[ \prod_{j \in J \cup \{i\}} x_j \right]   }{\mathbb{L}[x_i]}~~\qquad~~\mathbb{L}_{i,(0)}\left[\prod_{j \in J} x_j\right]~:=~\frac{\mathbb{L}\left[ \prod_{j \in J} x_j \right] - \mathbb{L}\left[ \prod_{j \in J \cup \{i\}} x_j \right]   }{1-\mathbb{L}[x_i]}.
		\end{align*}}
		Then it holds that $\mathbb{L}[\cdot]=\mathbb{L}[x_i]  \mathbb{L}_{i,(1)}[\cdot]   +  (1-\mathbb{L}[x_i])  \mathbb{L}_{i,(0)}[\cdot]  $. Moreover, both operators $\mathbb{L}_{i,(0)},\mathbb{L}_{i,(1)}$ are feasible for $\Prog{\sSAD}{\constraints}{2d-2}$.
		\end{lemma}
	The proof can be found in e.g.,~\cite[Lemma 2]{Rothvoss13}. Note that for every $i \in [n]$ and $\mathbb{L}_{i,(0)},\mathbb{L}_{i,(0)}$ satisfying the requirements in \cref{lem:SA_conditioning} we have $\mathbb{L}_{i,(0)}[x_i]=0,~~\mathbb{L}_{i,(1)}[x_i]=1$. 
	\cref{lem:SA_conditioning}, applied iteratively implies for every set $S\subseteq [n]$ with $|S| \leq d$ that every linear operator $\mathbb{L}[\cdot]$ feasible for $\Prog{\sSAD}{\constraints}{2d+2 \lceil d_{\constraints}/2 \rceil}$ can be written as a convex combination of linear operators, feasible for $\Prog{\sSAD}{\constraints}{2 \lceil d_{\constraints}/2 \rceil}$, that maps variables with indices in $S$ to 0 or 1. In other words, $\mathbb{L} = \sum_{i=1}^{2^d} \mathbb{L}_i$ such that for every $i$ and $j \in S$ we have $\mathbb{L}_i[x_j]\in\{0,1\}$.
	
	\begin{example}
	Consider a set of polynomials $\constraints =\{1,\pm(x_1^2-x_1),\pm(x_2^2-x_2), 3/2-x_1-x_2)\}$. The feasibility set of $\Prog{\sSAD}{\constraints}{2d+2 \lceil d_{\constraints}/2 \rceil}$ is the convex hull of its integral solutions.The feasibility set for $\Prog{\sSAD}{\constraints}{6}$ is shown in \cref{fig:SA_example} a. The feasibility set for a standard relaxation (replacing the integrality constraints with constraints $0\leq x_i \leq 1$) that corresponds to the feasibility region of $\Prog{\sSAD}{\constraints}{2}$ can be seen in \cref{fig:SA_example} b. The feasibility set of linear operators $\mathbb{L}[\cdot]$ feasible for $\Prog{\sSAD}{\constraints}{2}$ that additionally satisfies the property of being expressed as a convex combinations of operators integral on variable $x_1$ ($x_2$) can be seen in \cref{fig:SA_example} c (d), respectively. Finally, the set of operators feasible  for $\Prog{\sSAD}{\constraints}{4}$ can be seen in \cref{fig:SA_example} e. Note that the feasibility region in \cref{fig:SA_example} e is an intersection of regions in \cref{fig:SA_example} c and \cref{fig:SA_example}.
	\end{example}
	
	\begin{figure}[t]
		\includegraphics[width=1\linewidth]{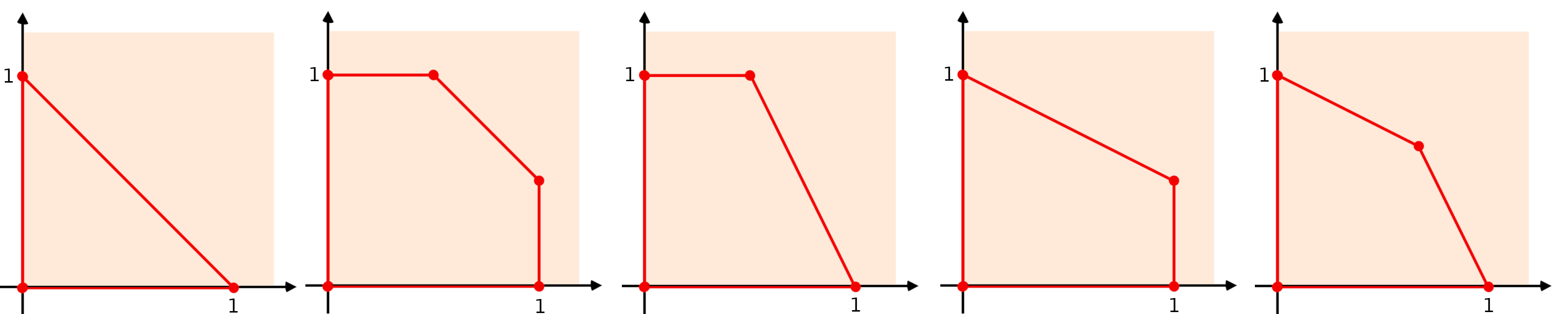} 
		\caption{Visualization of a feasibility set for different levels of $\PSAD$ and the \emph{conditioning} property from \cref{lem:SA_conditioning}.}
		\label{fig:SA_example}
	\end{figure}

	In the construction of the algorithm the conditioning property allows to take the solution of $\Prog{\sSAD}{\constraints}{2n+2 \lceil d_{\constraints}/2 \rceil}$ and choose a crucial variable that was assigned a fractional value and ask this value to be either 0 or 1, depending of the users preferences. The resulting solution is conditioned to be integral on this variable and stays feasible for $2n+2 \lceil d_{\constraints}/2 \rceil-2$ degree $\PSAD$. Clearly, the more variables are conditioned to be integral, the higher is the required degree of the $\PSAD$ solution. 

In general the conditioning property might not be easy to prove for a dual formulation of a given proof system. Assume that $\sGEN \subset \NNcone{\feasible}$ is such that the corresponding conic program admits no duality gap (the minimum value of the primal problem equals the maximum value of the dual program). Note that this assumption is necessary since~\cref{thm:General_proof_as_strong_as_SA} provides arguments on the primal side, namely the semialgebraic proof system. If the corresponding conic program admits a duality gap, then the implications of~\cref{thm:General_proof_as_strong_as_SA} on the dual side might not be correct.

 In what follows we provide a corollary giving sufficient conditions for the proof system to admit the conditioning property.

	\begin{corollary}
    Every proof system $\PGEN$ satisfying the requirements of \cref{thm:General_proof_as_strong_as_SA} admits the \emph{conditioning} property stated in \cref{lem:SA_conditioning}.  
    \end{corollary}

	\subsection{SONC vs SA}
		\label{sec:sonc_vs_sa_0/1}
	In this section we show two results. First, we show that $\TSONC^*$ is at least as strong as $\TSA^*$, meaning that SONC$^{~*}$ contains SA$^{~*}$. Second, by showing that every circuit polynomial is a $d$-junta, we show that $\TSA$ is at least as strong as $\TSONC$ and the same holds for systems strengthened with a Schm\"udgen-like Positivstellensatz i.e., SA contains SONC and SA$^{~*}$ contains SONC$^{~*}$. As a result we get that SA$^{~*}$ and SONC$^{~*}$ are polynomially equivalent.
	
	\begin{lemma}
		\label{lem:sonc_vs_sa_0/1}
		There exists a $c \in \N$ such that for every $d \in \N^*$ we get $\BoundSch{f}{\PSA}{\constraints'}{2d} \leq \ \BoundSch{f}{\PSONC}{\constraints'}{2cd}$, meaning that SONC$^{~*}$ contains SA$^{~*}$.
		\label{lemma:SONCatleastasstrongasSA}
	\end{lemma}
	\begin{proof}
		By \cite[Theorem 4.7]{DresslerKW18} the $\TSONC^*$ satisfies the condition of \cref{thm:General_proof_as_strong_as_SA} and thus the proof follows.
	\end{proof}
	Next we prove that the inverse of the inequality in \cref{lemma:SONCatleastasstrongasSA} also holds over the Boolean hypercube. We start with a technical lemma.

	\begin{lemma}
        Let $f$ be a circuit polynomial of total degree $d$. Then $f$ is a $d$-junta. 
        \label{Lemma:CircuitsAreJuntas}
	\end{lemma}
	\begin{proof}
		Consider a nonnegative circuit polynomial of the form
		\begin{align*}
		f(\mathbf{x}) \ := \ f_{\boldsymbol{\beta}} \mathbf{x}^{\boldsymbol{\beta}} + \sum_{j=0}^r f_{\boldsymbol{\alp}(j)} \mathbf{x}^{\boldsymbol{\alp}(j)}.
		\end{align*}	
		Since, by the \cref{Def:CircuitPolynomial}, $\boldsymbol{\beta} = \sum_{j = 0}^r \lambda_j \boldsymbol{\alp}(j)$ with $\sum_{j = 0}^r \lambda_j = 1$, $\lambda_j > 0$ and $\boldsymbol{\beta} \in \N^r$ we have $\lambda_j \leq \frac{1}{r+1}$ for at least one $j$. Thus, the total degree of $f$ satisfies $\deg(f) \geq r+1$. 
		
		This means on the contrary, if $f$ is of degree $d$, then the homogenization of $f$ contains at most $d$ many variables, which implies that $f$ is a $d$-junta by definition, see \cref{SubSec:PreliminariesSheraliAdams}.
	\end{proof}
	
	We obtain the following corollary:
	\begin{corollary}
		\label{cor:sonc_vs_sa_0/1}
        Let $f \in \R[\Vector{x}]$. For every $d \in \N^*$ we have $\Bound{f}{\PSONC}{\constraints}{2d} \leq \ \Bound{f}{\PSA}{\constraints}{4d}$ and $\BoundSch{f}{\PSONC}{\constraints}{2d} \leq \ \BoundSch{f}{\PSA}{\constraints}{4d}$. In other words SA contains SONC and SA$^{~*}$ contains SONC$^{~*}$.
	\end{corollary}
	
	\begin{proof}
        Assume that there exists a degree-$2d$ SONC certificate for $f$ of the form
        	$f = \sum_{j = 1}^m s_i g_i$
        such that the $s_i$ are SONCs and $g_i \in \constraints$. For every $i \in [m]$ we write $s_i = \sum_{j = 1}^{k_i} c_{ij}$ where $k_i \in \N$ and every $c_{ij}$ is a nonnegative circuit polynomial satisfying $\deg(c_{ij}g_i) \leq 2d$. By \cref{Lemma:CircuitsAreJuntas}, every $c_{ij}$ is a $\deg(c_{ij})$-junta, which yields the first inequality. For the second inequality the proof works analogously with $g_i \in \preprime{\constraints}$. 
	\end{proof}
	
	\subsection{SONC vs. SDSOS}
		\label{sec:sonc_vs_sdsos_0/1}
	In \cref{Sec:SDSOSvsSONC}, we saw that in general every SDSOS certificate is a SONC certificate, but not vice versa. Here, we show that the situation is more special on the Boolean hypercube: It turns out that the relation of the two certificates depends on the type of hierarchy, which we allow for SDSOS. We show the following result:
	
	\begin{theorem}
		\label{thm:sonc_vs_sdsos_0/1}
        Let $\constraints$ be as before with $\feasible \subseteq \{0,1\}^n$. Then for all $d \in \N$ we have the following dependencies $\HierSch{\PSDSOS}{\constraints}{2d} = \HierSch{\PSONC}{\constraints}{2d}$;  $\HierPut{\PSDSOS}{\constraints}{2d} \subsetneq\HierSch{\PSONC}{\constraints}{2d}$;  $\HierPut{\PSDSOS}{\constraints}{2d} \subseteq\HierPut{\PSONC}{\constraints}{2d}$. This implies that SONC$^{~*}$ is polynomially equivalent with SDSOS$^{~*}$, SONC$^{~*}$ strictly contains SDSOS and SONC contains SDSOS. 
	
	\end{theorem}
	
	\begin{proof}
			Following~\cite{DresslerKW18}[Definition 10], for every $\mathbf{v} \in \{0,1\}^n$ the function
		\begin{eqnarray}
		\struc{\delta_\mathbf{v} (\mathbf{x})} \ := \ \prod_{j \in [n]:~v_j=-1} \left( 1-x_j     \right) \cdot
		\prod_{j \in [n]:~v_j=1}  x_j    \label{Equ:DefKroneckerDelta}
		\end{eqnarray}
		is called the \struc{\emph{Kronecker delta (function)}} of the vector $\mathbf{v}$.
        By \cite[Theorem 12]{DresslerKW18} $f$ has a certificate of the form
        {\small	\begin{eqnarray}
\hspace*{1cm} f(\mathbf{x}) & = & \sum_{\mathbf{v} \in \feasible}  c_\mathbf{v} \delta_\mathbf{\mathbf{v}}(\mathbf{x})   +    \sum_{\mathbf{v} \in \{0,1\}^n \setminus \feasible}  c_\mathbf{v}  \delta_\mathbf{\mathbf{v}}(\mathbf{x}) p_\mathbf{v}(\mathbf{x})    +    \sum_{j = 1}^n s_j(\mathbf{x}) g_j(\mathbf{x}) + \sum_{j = 1}^n s_{n + j}(\mathbf{x}) (-g_j(\mathbf{x})),\label{Equ:PolynomialNonnegativityCertificate}
\end{eqnarray}}
	where $s_1,\ldots, s_{2n}$ are SONCs of degree at most $n-2$, $c_\mathbf{v} \in \mathbb{R}_{\geq 0}$, $g_i \in \preprime{\constraints}$, and $p_\mathbf{v} \in \constraints$.
	
	Part (1): By \cite[Lemma 14]{DresslerKW18}, $\delta_{\Vector{v}}(\Vector{x}) \in \preprime{\constraints}$. Thus, by \cref{cor:SONCcontainsSDSOS}, it only remains to show that every SONC involved in the certificate is a binomial square. But this was shown in Case 2 of the proof of \cite[Theorem 16]{DresslerKW18}.
	
	Part (2): Follows immediately from the fact that the example from \cite[Theorem 19]{DresslerKW18} to prove \cref{Corollary:ThereExistsNoPutinarPositivstellensatzForSDSOS} is an example defined over the Boolean hypercube.
	
	Part (3): Follows from the fact that, by \cref{Lemma:CircuitsAreJuntas}, every SDSOS certificate can be rewritten as the SONC certificate. 
	\end{proof}

	\subsection{SDSOS vs SA}
	\label{sec:sa_vs_sdsos_0/1}
	

	In this section we show that for BCPOPs, SA contains SDSOS.


	\begin{theorem}
		\label{thm:SDSoSvsSA_0/1}
        Let $\constraints$ be as before with $\feasible \subseteq \{0,1\}^n$. Let $f \in \R[\Vector{x}]$. For every $d \in \N^*$ we have $\Bound{f}{\PSDSOS}{\constraints}{2d} \leq \ \Bound{f}{\PSA}{\constraints}{4d}$.
        		
	\end{theorem}
	
	The theorem follows immediately from \cref{cor:sonc_vs_sa_0/1} and \cref{thm:sonc_vs_sdsos_0/1}. We provide, however, an independent proof here, which works on the dual side, and which we consider to be individually interesting.

	For every set $K \subseteq [n]$ we denote its \struc{\textit{power set}} by $\struc{\mathcal{P}(K)}$. 

   	\begin{proof}
   		Consider the problem
   		\begin{equation}
   		\label{eq:SDSoSvsSA_SA_2}
   		\struc{\PSAD^{{'}4d}}:\min\left\{\tilde{\mathbb{E}}[f]~ \middle|~ \tilde{\mathbb{E}}[1] = 1,~ {M_g^{4d}}_{\big|\mathcal{P}(K)}   \succeq 0,~  g \in \mathcal{G},~ K \subseteq [n],~  |K|\leq  \left\lfloor \frac{4d-\deg(g)}{2} \right\rfloor \right\}.
   		\end{equation}  
   		This problem is a relaxation of the problem $\PSOSD^{4d}$, see~\eqref{eq:intro_PP_6}, since, by Sylvester's Theorem, a symmetric matrix $M$ is PSD if and only if all the principal submatrices are PSD. 
   		
   		First, we show that the degree $4d$ Sherali Adams is equivalent to~\eqref{eq:SDSoSvsSA_SA_2}, a similar reasoning can be found also e.g., in~\cite[Section 3.2, Equation (19)]{Laurent03}. Consider a constraint $g \in \mathcal{G}$ and a set $K \subseteq [n]$, such that $|K| \leq \lfloor\frac{4d-\deg(g)}{2} \rfloor$. A \struc{\emph{M\"obius matrix} $Z_K^{-1}$}, $Z_K^{-1}\in \{-1,0,1\}^{2^{|K|} \times 2^{|K|}}$, is a square matrix indexed by subsets $I,J$ of $K$ such that $Z_K^{-1}(I,J)=(-1)^{|J \setminus I| }$ if $I \subseteq J$ and $Z_K^{-1}(I,J)=0$ otherwise. Compute a matrix:
   		\begin{equation}
   		\struc{{D_g^{4d}}_{\big|\mathcal{P}(K)}} \ := \ Z_K^{-1}     {M_g^{4d}}_{\big|\mathcal{P}(K)}    \left( Z_K^{-1} \right)^\top.
   		\end{equation}
   		One can check that ${D_g^{4d}}_{\big|\mathcal{P}(K)}$ is a diagonal matrix with entries 
   		$$
   		{D_g^{4d}}_{\big|\mathcal{P}(K)}(I,I) = \sum_{I \subseteq H \subseteq K} (-1)^{|H \setminus I|} \tilde{\mathbb{E}}\left[ g \cdot  \prod_{x \in H} h_h \right] = \tilde{\mathbb{E}}\left[ g\cdot  \prod_{i \in I} x_i \prod_{j \in K \setminus I} \left( 1- x_j \right)      \right].
   		$$
   		see. e.g.~\cite[Lemma 2]{Laurent03}.

   		Finally, since $Z_K^{-1}    {M_g^{4d}}_{\big|\mathcal{P}(K)}    \left( Z_K^{-1} \right)^\top$ is a congruent transformation of ${M_g^{4d}}_{\big|\mathcal{P}(K)}$ that preserves eigenvalues, we get that ${M_g^{4d}}_{\big|\mathcal{P}(K)}  \succeq 0 \Leftrightarrow {D_g^{4d}}_{\big|\mathcal{P}(K)}   \succeq 0 \Leftrightarrow  {D_g^{4d}}_{\big|\mathcal{P}(K)}(I,I) \geq 0$, for all $I \in \mathcal{P}(K)$. Setting $J =K \setminus I$ we get a one-to-one mapping to functions in $\PSA_\mathcal{G}^{4d}$.
   		
   		Next, we show that every linear operator $\tilde{\mathbb{E}}[\cdot]$ that is feasible for $\PSAD^{{'}4d}$ is also feasible for $\PSDSOSD^{{'}2d}$, that is, for every  $g \in \mathcal{G}$ and every $K,~L \subseteq [n]$, such that $|K|,~|L| \leq \lfloor \frac{2d-\deg(g)}{2} \rfloor$ the matrix ${M_g^{2d}}_{\big|\{K,L\}} $ is PSD.

   		
   		 Consider the set $H =K \cup L$. Since $|H| \leq 2 \lfloor \frac{2d-\deg(g)}{2} \rfloor$, for $d,\deg(g) \in \mathbb{N}_+$ we have $|H| \leq \lfloor \frac{4d-\deg(g)}{2} \rfloor$. Now consider a submatrix ${M_g^{4d}}_{\big|\mathcal{P}(H)} $. Every linear operator $\tilde{\mathbb{E}}[\cdot]$ that is feasible for $\PSAD^{{'}4d}$ has to satisfy ${M_g^{4d}}_{\big|\mathcal{P}(H)} \succeq 0 $. Finally, since ${M_g^{2d}}_{\big|\{K,L\}} $ is a principal submatrix of 
   		${M_g^{4d}}_{\big|\mathcal{P}(H)} $, by the Sylvester's criterion the linear operator $\tilde{\mathbb{E}}[\cdot]$ that is feasible for $\PSAD^{{'}4d}$ has to satisfy also ${M_g^{2d}}_{\big|\{K,L\}} \succeq 0 $. This finishes the proof.
   	\end{proof}

	
%
%
	\bibliographystyle{amsalpha}
	\bibliography{Hier_cont}
	
	
	

\end{document}